\documentclass[journal]{IEEEtran}
\usepackage{times,amsmath,amssymb,float,enumerate,bm,graphicx}
\usepackage{amsthm}
\usepackage{color}

\newtheorem{Prop}{Proposition}

\def\C{{\bf C}}
\def\D{{\bf D}}

\def\F{{\mathbf F}}
\def\K{{\mathbf K}}
\def\G{{\mathbf G}}

\def\I{{\mathbf I}}
\def\S{{\bf S}}
\def\X{{\bf X}}
\def\x{{\bf x}}
\def\T{{\bf T}}

\def\U{{\bf U}}
\def\V{{\bf V}}
\def\W{{\bf W}}
\def\Y{{\bf Y}}
\def\Z{{\bf Z}}
\def\u{{\bf u}}
\def\v{{\bf v}}

\def\BH{{\bf H}}
\def\BF{{\bf F}}

\def\BG{{\bf G}}
\def\I{{\bf I}}
\def\R{{\bf R}}
\def\P{{\bf P}}
\def\Q{{\bf Q}}
\def\z{{\bf z}}

\def\x{{\bf x}}

\def\y{{\bf y}}

\def\0{{\bf 0}}

\definecolor{ao}{rgb}{0.0, 0.5, 0.0}

\definecolor{brightmaroon}{rgb}{0.76, 0.13, 0.28}

\title{%Explorations of 
Testing for  Causal Influence using a Partial Coherence Statistic}

\author{\IEEEauthorblockN{Louis Scharf\IEEEauthorrefmark{1}, ~\IEEEmembership{Fellow,~IEEE}, and
		Yuan Wang\IEEEauthorrefmark{2}},
	\IEEEauthorblockA{\IEEEauthorrefmark{1}Departments of Mathematics and Statistics, Colorado State University, 
		Fort Collins, CO, USA}
		\IEEEauthorblockA{\IEEEauthorrefmark{2}Department of Mathematics and Statistics, Washington State University, Pullman, WA, USA}
	\thanks{This work was supported in part by   the National Science Foundation
		under contract  CCF-1712788. 
		Corresponding author:Y. Wang (email: yuan.wang@wsu.edu).}}

\begin{document}
	
	% The paper headers
	\markboth{IEEE Transactioin in Signal Processing,~Vol.~XXX, No.~X, May~2021}%
	{Shell \MakeLowercase{\textit{et al.}}: Bare Demo of IEEEtran.cls for IEEE Journals}
\maketitle

%
%%
%%%%%%%%%%%

\begin{abstract}
\noindent In this paper we explore partial coherence as a tool for evaluating  causal influence of one signal sequence on another. In some cases the  signal sequence is sampled from a  time- or space-series. The key idea is to establish a connection between  questions of causality and questions of partial coherence. Once this connection is established, then a scale-invariant partial coherence  statistic is used to resolve the question of causality. This coherence statistic is shown to be a likelihood ratio, and its null distribution is shown to be a Wilks Lambda. It may be computed from a composite covariance matrix or from its inverse, the information matrix. Numerical experiments demonstrate the application of partial coherence to the resolution of  causality. Importantly, the method is model-free, depending on no generative model for causality.\\

\end{abstract}

%\begin{keywords}
%\end{keywords}

\IEEEpeerreviewmaketitle

\section{Introduction}
Questions of causality are fraught with ambiguities. In fact, we might paraphrase Norbert Wiener's discussion of Newtonian physics to say  causality is a concept grounded in a physical picture that can never be fully justified or fully rejected experimentally \cite{Wiener}. For this reason it is important to state at the outset that any  conclusions are conditioned on available {\em prima facie evidence} and on the  {\em method of inference}. The emergence of new evidence may render conclusions moot or false. And, more sophisticated methods of inference may reveal dependencies not revealed by other methods. For this reason it is prudent to answer questions of causality in the negative, rather than the affirmative. In other words, statements like  $x$ {\em is not} a cause for $y$ are preferred over statements like $x$ {\em is}  a cause for $y$. Moreover, a conclusion based on data is perforce a conclusion based on an experiment and a statistic whose value is used to resolve the question of causality. This means the question can only be resolved with respect to prima facie evidence, the outcome of an experiment, and the value of a statistic. It follows that the question is only resolved at a prescribed significance level.  Change the significance level, and change the answer to the question of causality. It might be said that the  question of causality is practically a question of null hypothesis testing, where the null hypothesis states that one signal sequence has no causal influence on another. %there is no causal dependence of one  measurement on another.  
One may still speak of the power of a method to reject the null for a putative causal model, but this does not validate the  causal model, nor does it rule out the possibility that other methods and evidence would negate the finding of causality.

Granger causality \cite{Granger:1969, Granger:1980} has been widely used for addressing questions of causality. Loosely speaking, a time series $\{y_n\}$ is causally dependent on another time series $\{x_n\}$ if the past of $\{x_n\}$ contains unique information about $\{y_n\}$. This statement is made precise in terms of joint probability statements concerning the predictive value of a time series when composed with prima facie evidence.  Geweke \cite{Geweke82} has studied the linear dependence and feedback between two time series and {in \cite{Geweke84} extended his results to conditional linear dependence in the presence of a third time series. Solo \cite{Solo} analyses the Granger-Geweke causality measure for a variety of state-space models, and establishes the invariance of this measure to linear time-invariant filtering.  Siggiridou and Kugiumtzis \cite{Siggiridou:2016} used a conditional Granger causality index %(CGCI) 
%is used in \cite{Siggiridou:2016} 
to quantify Granger causality in a vector autoregression. Kramer \cite{Kramer} and Amblard and Michel \cite{Amblard2011} define and analyze {\em directed information} in their study of Granger causality graphs. %, which in our study of partial coherence we have called conditional information rate. It is also conditional Kullback-Leibler divergence and conditional mutual information. 
%Our framing of questions of causality as tests of hypotheses leads to an ordinary likelihood statistic that is a sampled-data version of partial coherence, or directed information rate. } 
  Schamberg and Coleman \cite{Schamberg:Coleman:2020} recently proposed a novel sample path dependent measure of
causal influence between time series. % and leverages the sequential prediction theory  to estimate the causal measure. }

Barnett et al.  \cite{Barnett1, Barnett2, Barnett3} have adapted Granger's point of view to the study of causality in state-space models. They derive an estimate of the  Granger-Geweke causality measure} for testing causality within state-space models and  demonstrate   improved statistical power and reduced bias compared to results obtained by fitting an autoregressive (AR) model to what should be an autoregressive-moving average (ARMA) model.  

%{\color{blue} This Granger causality measure is our partial coherence, which we derive from reasoning about partial correlation. Our estimate of partial coherence is an ordinary likelihood ratio for null hypothesis testing in  a model-free, second-order, analysis of partial correlation. The estimate of \cite{Barnett1} is an {\em estimate and plug} estimate of partial coherence, based on  identification of an underlying state space model for the measurements. Our estimate is ordinary likelihood in a model-free analysis of causality, and theirs is  approximate likelihood in a state-space analysis of causality.}

%In practice, causality is  evaluated  with respect to the available information.  Let $\{J_n\}$ be the set of available information at time $n$. Then, the time series $\{x_n\}$ is said to be a  cause of $\{y_n\}$ with respect to  prima facie evidence $\{J_n\}$ if  knowledge about $\{y_n\}$ is increased by adding $\{x_n\}$ into the information set $\{J_n\}$. 
In this paper, we are interested in assessing %temporal 
{\em causal  influence} between two  signal sequences, %vector-valued  measurements %time series, 
using partial correlation or {\em partial coherence}.  It seems to us that {\em causal influence} is more descriptive than {\em causality} for describing the aims of statistical methods that are applied to  signal sequences %measurements, 
without any %presence 
 pretense at identifying or verifying a physical mechanism. As a practical matter of testing, these  signal sequences %  measurements %time series are replaced by 
are often finite-dimensional samples of two time series, %and 
represented as the random vectors $\x$ and $\y$. %Partial correlation coefficients measure correlation between two estimator errors for random variables that have been regressed onto a third set of variables. 
The partial correlation between  $\x$ and $\y$, given $\z$,  is the correlation between the residual for  $\x$ regressed on $\z$ and the residual for $\y$ regressed on $\z$, where $\z$ is a third set of  random variables.  A zero partial correlation indicates that the random vector $\y$ does not participate in a linear minimum mean-squared error estimator of $\x$ from $(\y,{\bf z})$. In the multivariate normal case, $\x$ is then conditionally independent of $\y$, given $\z$.  We shall  define the prima facie evidence $\z$  in such a way that  the question of causality can be resolved from an analysis of partial coherence. 
%We show that a partial coherence statistic is a scale-invariant 
%measure of causality that has a connection to likelihood testing of hypotheses. It determines Kullback-Leibler divergence, and the conditional rate at which $\y$ brings information about $\x$, given $\z$.\\

In  a comprehensive treatment of conditioning and partial correlation for vector-valued random variables, with no reference to the question of causality, we identify a quantity we call {\em partial coherence}. %establish a connection between conditional mutual information, conditional Kullback-Leibler divergence, conditional information rate and a quantity we have called {\em partial coherence}. 
Partial coherence is a ratio of determinants that is proportional to the volume of a normalized  error concentration ellipse for linearly estimating a pair of random variables from a third, or equivalently  for linearly estimating one random variable from two others. It is determined by partial canonical correlations.   Partial coherence is invariant to block diagonal non-singular transformations. %, a result analogous to Solo's finding for time series.  
In the case of wide-sense stationary time series, the spectral representation of partial coherence shows it to be a Riemann sum of narrowband coherences or information rates, in close analogy with Geweke's spectral representations. With appropriate choice of conditioning,   partial coherence is a monotone function of  the Granger-Geweke causality measure,  which is the same as 
Geweke's transfer entropy. 

We then turn to the question of causality, following the language and philosophy of Granger, and studying problems of the same general class as the Barnetts and Seth.  However, in our study of causality between  signal sequences. %time- or space-series, 
we adopt a conservative, model-free, characterization of the signals %time series 
to be analyzed, based only on second-order moments. That is, in contrast to the work of the Barnetts and Seth, where parametric state space models are assumed for the time series under study, we assume only that samples of a time series are second-order random variables.  We show that the sampled-data estimate of partial coherence is in fact the model-free {\em ordinary likelihood ratio} for testing the null hypothesis that one signal sequence %time series 
is not correlated with another. %causally dependent upon another. %Our sampled-data estimate of partial coherence is in fact the ordinary likelihood ratio for null hypothesis testing in a model-free, second-order, analysis of causality.  
The  statistic %estimate 
of \cite{Barnett1} is an approximate likelihood ratio, based on an approximate ML identification of an underlying state space model for the  signals. %measurements. 
That is, our estimate is ordinary likelihood in a model-free analysis of causality, and theirs is  approximate likelihood in a state-space analysis of causality. Importantly, under the null hypothesis, ordinary likelihood, or estimated partial coherence, is shown to be distributed as a Wilks Lambda statistic, allowing for the setting of thresholds to control significance level in hypothesis testing.

In the examples we analyze, the signal sequences %time series 
of interest are finite samples from wide-sense stationary time series, but this is not essential to our methods. In fact the methodology we advocate extends  to the study of space-time series, as in \cite{Ramirez2010}, where the methods of this paper may be used to test for particularly influential time series in a spatially-distributed set of time series. Even more generally, partial coherence may be used to characterize causality between quite arbitrary measurement vectors, as in the labeling of graph edges by Whittaker~ \cite{Whittaker:2009}.

When parametric models are accurate, and the ratio of sample count to parameter count is large, then one would expect parametric methods to outperform model-free methods. But when the parametric model is inaccurate, model bias is introduced. %, and when parameter counts are numerous with respect to sample counts, model-free methods are preferred. 
In our re-running of the Barnetts and Seth example, we find that the statistical power we achieve is generally comparable to theirs, without using a parametric state-space model. We use a receiver operating characteristic (ROC) %into the literature of causality, in order
 to emphasize the point that resolution of causality is always a problem in hypothesis testing, where the power to identify (provisional) causality is dependent upon the significance level of the test, or equivalently upon the  probability of incorrectly rejecting the hypothesis of  non-causality.

% Previous work in \cite{Geweke82}-\cite{Barnett3} have studied the time-domain Granger causality among multiple time series. The causality from a time series $\{y_t\}$ to  $\{x_t\}$ with respect to a third time series $\{z_t\}$ is defined by the predictive power gained by adding the past of $\{y_t\}$ in addition to the past of $\{x_t\}$  and $\{z_t\}$. More specifically, the causality is measured by the logarithm ratio of determinant  for  $\Sigma_{xx}$ over  $\Sigma_{xx}^-$ where  $\Sigma_{xx}$ is the error covariance for predicting $\x_t$ using the past of $\x_t$, $\z_t$, and $\y_t$, and  $\Sigma_{xx}^{-}$ is with past of $\y_t$ excluded. In the null case when $\{y_t\}$ has no causal impact on $\{x_t\}$, the log ratio is zero.  Auto-regressive models have been used to estimate the test statistic.  Barnett et al. \cite{Barnett1} proposed a more general state space modeling approach for estimating the causality statistic and had shown in the paper  improved statistical power and reduced bias comparing to standard autoregressive methods.  

\textbf{Notations:} Bold upper
case letters denote matrices, boldface lower case letters denote column vectors, and italics denote scalars. The superscript $(\cdot)^T$ %and $(\cdot)^H$ 
denotes transpose. % and Hermitian, respectively. 
${\bf I}_M$ is the identity matrix of dimensions $M \times M$, and ${\bf 0}$ denotes either a zero column vector  or the zero matrix of appropriate dimensions. %(the difference should be clear from the context).  
 We use ${\bf A}^{1/2}$ (${\bf A}^{-1/2}$) to denote the square root matrix of the Hermitian matrix ${\bf A}$ (${\bf A}^{-1}$); $ {\rm blkdiag}[{\bf A}_1, \ldots, {\bf A}_k ]$ is a block-diagonal matrix whose diagonal blocks are ${\bf A}_1, \ldots, {\bf A}_k$.  The notation ${\bf x} \sim \mathcal{N}_M({\bf 0}, {\bf R})$ indicates that ${\bf x}$ is an $M$-dimensional %complex circular 
Gaussian random vector of mean ${\bf 0}$ and covariance ${\bf R}$.  We say it is multivariate normal (MVN).  

The rest of the paper is organized as %following. 
 follows. In Section ~\ref{sec:PredCausal} we review a few motivating notions of predictive causal influence. In Section~\ref{sec:pc}, we introduce partial coherence and demonstrate its usefulness in the estimation of one random vector from two others. We review how partial coherence may be extracted from the information matrix, and establish spectral formulas which bring insight into the way  partial canonical correlations additively decompose partial coherence. %In Section~\ref{sec: KLdiv}, we 
We establish,  in the case of multivariate normal experiments, the connection  between partial coherence, conditional Kullback-Leibler divergence, conditional mutual information, conditional  information rate, and a Granger-Geweke causality measure. % in the case of multivariate normal experiments. 
In Section~\ref{sec:test} we %establish the connection between partial coherence and questions of causality, and 
show that partial coherence is  a scale-invariant likelihood ratio statistic in the case of multivariate normal experiments. Importantly, we establish that the null distribution of the  coherence statistic is the distribution of a Wilks Lambda statistic, and show that its stochastic representation is that of a product of independent beta-distributed random variables. %In Section~\ref{sec:ts}, we review  a few standard ways of talking about causal influence, and establish the connection with  linear prediction in the case of Gaussian time series. 
In Section~\ref{sec:exp}, we conduct two experiments to explore the past-causal, future-causal, and mixed-causal relationship between time series.  The first of these experiments demonstrates the use of {\em pairwise} partial coherence as a map of causality %versus  
as a function of two time variables. The second reproduces the experiment of Barnett and Seth \cite{Barnett1} in order to compare our model-free  estimator of causality to their model-based estimator. We compare the power of our  estimator  of partial coherence  to the power of their estimator, and demonstrate the use of the Receiver Operating Characteristic as a way to view the trade-off between power (probability of detection, or one minus type II error probability) versus size (significance level, probability of false alarm, or probability of type I error).
%Section~4 discusses the connection between this paper and prior work of Geweke \cite{Geweke82,Geweke84} and by Barnett, et al. \cite{Barnett1, Barnett2}.\\ 
 In Section~\ref{sec:discuss} we offer concluding remarks on the efficacy of model-based versus model-free approaches to null hypothesis testing for  causal influence. % in Section~6.

%
%%
%%%%%%%%%%

\section{Notions of Predictive Causal Influence}~\label{sec:PredCausal}
Begin with two scalar-valued time series, whose finite 
pasts may be denoted $x^t=(x_0, x_1, \ldots, x_t)$ and $y^t=(y_0, y_1, \ldots, y_t)$. Assume the existence of a joint probability density function (pdf) $p(x^t, y^t)$, for all $t$, with corresponding marginal pdfs for any subset of these random variables. A comparison between the conditional pdfs $p(y_{t+1}|x^t, y^t)$ and $p(y_{t+1}|y^t)$ %for $t\in T$ 
may be used to test the  null hypothesis ${\makebox H}_0$ that $x^t$ has no causal influence on $y_{t+1}$. Such a comparison might be based on the error variance $V_{y_{t+1}|x^t,y^t}$ of a  predictor of $y_{t+1}$ from $(x^t, y^t)$ versus the error variance $V_{y_{t+1}|y^t}$ of a predictor of $y_{t+1}$ from $y^t$ only. There are other possibilities. Conditioned on $y^t$, are $y_{t+1}$ and $x_s$ correlated? When this question is asked for $s=t\pm r$, then hypotheses  of past-causal, future-causal,  mixed-causal influence may be answered.  These questions are easily generalized to  multivariate time series by comparing the {\em error covariance} for predictors of %the 
a random vector $\y_{t+1}$ from random vectors $(\x^t,\y^t)$. %, where $\y$ is extracted from $\y^t$, and $(\x,\z)$ are extracted from the future and past of $\x^t$. This generalization leads naturally to the study of partial coherence.}

When prediction error variances are replaced by  entropies, %${\mathbb E}[-\log p(x_{t+s})], {\mathbb E}[-\log(x_{t+s}|x^t)]$, and  ${\mathbb E}[-\log(x_{t+s}|x^t, y^t)]$, 
then entropy differences may be used to measure mutual information. Define the following entropies $H$ and mutual information $I$: %predictive variances are replaced by mutual informations $I$: %, which are special cases of Kulback-Leibler divergences $K$:
\begin{align}
	&H_{y_{t+s}}= {\mathbb E}[ -\log p(y_{t+s})], \,\makebox{entropy of the random variable }  y_{t+s},\nonumber \\
	&H_{y_{t+s}|y^t}={\mathbb E}[ -\log p(y_{t+s}|y^t)],\nonumber \\
&	 H_{y_{t+s}|x^t,y^t}={\mathbb E}[ -\log p(y_{t+s}|x^t, y^t],\nonumber \\
&	I_{y_{t+s},x^t| y^t}=H_{y_{t+s}|y^t}-H_{y_{t+s}|x^t,y^t}. \nonumber%&%KL[p(x_{t+s}, x^t)||p(x_{t+s})p( x^t)]=
	%{\mathbb E}[\log \frac{p(x_{t+s}, x^t)}{p(x_{t+s})p( x^t)}]
	%=H_{x_{t+s}}-H_{x_{t+s}|x^t}
\end{align}
In these expressions, the random variable $p(y_{t+s})$ is a pdf random variable, and ${\mathbb E}$ is expectation. The entropy  $H_{y_{t+s}}$ is a measure of uncertainty for the random variable. %, and it determines the code-length for lossless encoding. 
The other entropies are entropies conditioned on the past $y^t$ or the two pasts $(x^t,y^t)$. The mutual information $I_{y_{t+s},x^t| y^t}$ measures mutual information between $y_{t+s}$ and $x^t$, after conditioning on the past $y^t$. A special case is $I_{y_{t+s},x_s| y^t}$, which measures the mutual information between  $y_{t+s}$ and $x_s$, after conditioning on the past $y^t$. 
By construction the conditional entropies and the conditional mutual information  are  measures of {\em directed information}. For ergodic time series, %$I_{x_{t+1}, x^t}$ 
they may be averaged over $t$ to determine the minimum codelength for encoding the time series $(y_0, y_1, \ldots)$, with or without the side information $x^t$. %In the introduction we have paraphrased Wiener regarding the very notion of causality. 
 Rissanen and Wax \cite{Rissanen:Wax:1987} are 
as cautionary as Wiener in their use of the term causality when they say, ``We measure causality by predictability, with predictability interpreted in terms of the codelength required to encode the time series with predictive coding. ... %And in their conclusions they write, ``
Obviously, the value of the resulting measure depends critically on the class of predictive densities selected. A good selection gives a sharp measure of causal dependence while a bad one masks a possible causal dependence, which of course is just as it should be."

In the framing of questions about past-causal, future-causal, or mixed causal-influence, three random vectors $(\x,\y,\z)$ may be %variables are 
extracted from two  signal sequences. %time series. 
For example, $\x=x^t, \y=y_{t+1}, \z=y^t$. Or $\x=x_s, \y=y_{t+1}, \z=y^t$.  More generally, $\x$ is a $p$-sample %of measurements 
from time series $\{x_n\}$, $\y$ is a $q$-sample %measurement 
from the time series $\{y_n\}$, and $\z$ is an $r$-sample %of measurements 
from either $\{x_n\}$ or $\{y_n\}$. That is, three finite-dimensional vectors are extracted from two signal sequences. The construction determines the nature of the question of causal influence. This suggests a general framework for analyzing causal influence based on three random vectors $\x\in {\mathbb R}^p, \y\in {\mathbb R}^q, \z\in {\mathbb R}^r$. Then partial coherence may be used as a tool for testing hypotheses of causal influence.  It encodes for conditional prediction error covariance. For zero-mean Gaussian time series, in which case all finite-dimensional random vectors are multivariate normal (MVN),  partial coherence encodes for %prediction  error covariances, 
conditional entropies, and mutual information. %It  reveals the role of conditioning in a particularly illuminating way. 
These points will be further clarified in due course. 

%Throughout the paper we assume multivariate normality. 
We do not impose a finite-dimensional time series model, such as AR, ARX, ARMA, or ARMAX  for signal sequences. Rather we assume only that the signal sequences %time series 
are finite-variance with arbitrary non-negative definite covariances for every finite windows-worth of the sequences. %series. 
These are then estimated from multiple realizations of an experiment, or in the case of ergodic times series, from consecutive windows of the signal sequences. %tme series. 
There is no parameter estimation for coefficients in a time series model that is based on a finite-dimensional difference equation for covariances. 

\section{ Predictive Causality, Linear Regression, and Partial Coherence %, Kullback-Leibler Divergence, and Conditional Information Rate
}\label{sec:pc} 
 There are two ways to frame questions of causal influence: %We shall frame questions of causality 
as  a question of estimating two random vectors from one other, which is to say regressing two random vectors onto one random vector; or as a question of estimating one random vector from two others, which is to say regressing one random vector onto two random vectors.  In the first case % The set up is this: 
 the random vectors $(\x,\y)$ are to be regressed onto the random vector $\z$; in the second case,  
the random vector ${\bf x}\in {\mathbb R}^p$ is to be regressed onto the random vectors ${\bf y}\in {\mathbb R}^q$ and ${\bf z}\in {\mathbb R}^r$. The composite covariance matrix for these three random vectors is
\begin{equation}\label{comp}
\R= {\mathbb E}[\left[ \begin{array} {ccc} {\bf x}\\ {\bf y}\\ {\bf z} \end{array} \right]\left [\begin{array}{ccc}{\bf x}^T & {\bf y}^T & {\bf z}^T \end{array} \right]]=\left[ \begin{array} {ccccccccc} {\bf \R_{\x\x}}&\R_{\x\y} & \R_{\x\z}\\ {\bf \R_{\y\x}}&\R_{\y\y} & \R_{\y\z}\\{\bf \R_{\z\x}}&\R_{\z\y} & \R_{\z\z}\end{array} \right].
\end{equation}
By defining the composite vectors $\u=(\x^T, \y^T)^T$, and $\v=(\y^T, \z^T )^T$ the covariance matrix $\R$ may be parsed two ways: 
\begin{equation}
\R=\left[ \begin{array}{cccc} \R_{\u\u} & \R_{\u\z}\\ \R_{\u\z}^T & \R_{\z\z}\end{array} \right]=\left[ \begin{array}{cccc} \R_{\x\x} & \R_{\x\v}\\ \R_{\x\v}^T & \R_{\v\v}\end{array} \right].
\end{equation}
 The covariance matrix $\R_{\u\u}$ is $(p+q)\times (p+q)$, and the covariance matrix $\R_{\x\x}$ is $p\times p$. %Consequently, there 
There are two useful representations for the inverse of the composite covariance matrix $\R$:
\begin{align}
\R^{-1}=&\left[ \begin{array}{cccc} \R_{\u\u|\z}^{-1} & \R_{\u\u|\z}^{-1}\R_{\u\z}\\ \R_{\u\z}^T\R_{\u\u|\z}^{-1} & \R_{\z\z}^{-1}+\R_{\z\z}^{-1}\R_{\u\z}^T\R_{\u\u|\z}^{-1}\R_{\u\z}\R_{\z\z}^{-1} \end{array} \right] \nonumber \\
=&\left[ \begin{array}{cccc} \R_{\x\x|\v}^{-1} & \R_{\x\x|\v}^{-1}\R_{\x\v}\\ \R_{\x\v}^T\R_{\x\x|\v}^{-1} & \R_{\v\v}^{-1}+\R_{\v\v}^{-1}\R_{\x\v}^T\R_{\x\x|\v}^{-1}\R_{\x\v}\R_{\v\v}^{-1} \end{array} \right].
 \label{invR} 
%=&\left[ \begin{array}{cccc} \R_{\x\x|\v}^{-1} & \R_{\x\x|\v}^{-1}\R_{\x\v}\\ \R_{\x\v}^T\R_{\x\x|\v}^{-1} & \R_{\v\v}^{-1}+\R_{\v\v}^{-1}\R_{\x\v}^T\R_{\x\x|\v}^{-1}\R_{\x\v}\R_{\v\v}^{-1} \end{array} \right] \nonumber
\end{align}
The matrix $\R_{\u\u|\z}$ is the error covariance matrix for estimating the composite vector $\u$ %=(\x^T, \y^T)^T
from $\z$, and the matrix $\R_{\x\x|\v}$ is the error covariance matrix for estimating $\x$ from $\v$: % the pair $\v=(\y^T,\z^T  )^T$:
\begin{align}
\R_{\u\u|\z}=\R_{\u\u}-\R_{\u\z}\R_{\z\z}^{-1}\R_{\u\z}^T, \\
\R_{\x\x|\v}=\R_{\x\x}-\R_{\x\v}\R_{\v\v}^{-1}\R_{\x\v}^T.
\end{align}
We shall have more to say about these error covariance matrices in due course. Importantly, the inverses of each may be read out of the inverse for the composite covariance matrix $\R^{-1}$ of eqn(\ref{invR}). The dimension of the error covariance $\R_{\u\u|\z}$ is $(p+q)\times (p+q)$ and the dimension of the error covariance $\R_{\x\x|\v}$ is $p\times p$.

\subsection{ Regressing two random vectors onto one: partial correlation and partial canonical correlation coefficients
}

The estimators of $\x$ and $\y$ from $\z$, and their resulting error covariance matrices are easily read out from the composite covariance matrix $\R$:
\begin{align}\label{filt}
\hat{\x}(\z)=&\R_{\x\z}\R_{\z\z}^{-1} \z, \\
\R_{\x\x|\z}=&{\mathbb E}[(\x-\hat{\x})(\x-\hat{\x})^T]=\R_{\x\x}-\R_{\x\z}\R_{\z\z}^{-1}\R_{\x\z}^T,\\
\hat{\y}(\z)=&\R_{\y\z}\R_{\z\z}^{-1} \z, \\
\R_{\y\y|\z}=&{\mathbb E}[(\y-\hat{\y})(\y-\hat{\y})^T]=\R_{\y\y}-\R_{\y\z}\R_{\z\z}^{-1}\R_{\y\z}^T .
\end{align}
The composite error covariance matrix for the errors $\x-\hat{\x}(\z)$ and $\y-\hat{\y}(\z)$ is the matrix 
\begin{align}\label{error}
\R_{\u\u|\z}=&{\mathbb E}\left[\begin{array}{cc} \x-\hat{\x}(\z)\\ \y-\hat{\y}(\z) \end{array}\right]\left[\begin{array}{cc} (\x-\hat{\x}(\z))^T& (\y-\hat{\y}(\z))^T\end{array} \right] \nonumber \\
=&\left[\begin{array}{cccc} \R_{\x\x|\z} &\R_{\x\y|\z}\\ \R_{\x\y|\z}^T & \R_{\y\y|\z} \end{array}\right], \\
\R_{\x\y|\z}=&{\mathbb E} [(\x-\hat{\x}(\z))(\y-\hat{\y}(\z))^T] \nonumber\\
=&\R_{\x\y}-\R_{\x\z}\R_{\z\z}^{-1}\R_{\y\z}^T.
\end{align} 
 In this matrix the matrix block $\R_{\x\y|\z}$ %=E[(\x-\hat{\x}(\z))(\y-\hat{\y}(\z))^T]=\R_{\x\y}-\R_{\x\z}\R_{\z\z}^{-1}\R_{\y\z}^T$ 
is the $p\times q$ matrix of cross correlations between the random errors $(\x-\hat{\x}(\z))$ and $(\y-\hat{\y}(\z))$. It is called the {\em partial correlation between the random vectors $\x$ and $\y$, after each has been regressed onto the common vector $\z$.}  Under the null hypothesis $\makebox{H}_0$,  this partial correlation is zero, and therefore  $\R_{\u\u|\z}$ is the block-diagonal matrix $\text{blkdiag}[\R_{\x\x|\z}, \R_{\y\y|\z}]$.

% $\makebox{H}_0: \R_{\x\y|\z}={\bf 0}$, therefore $\R_{\u\u|\z}^0=\text{blkdiag}[\R_{\x\x|\z}, \R_{\y\y|\z}].$  

	Formula \eqref{invR} %and \eqref{ } %give intuition for how the partial canonical correlations figure into the determination of normalized error covariance. But it is perhaps more important to note that the partial coherence and 
	shows that $\R_{\u\u|\z}$ and its determinant may be read directly out of the $(p+q)\times (p+q)$ Northwest block of the inverse of the error covariance matrix:
	\begin{align}
		(\R^{-1})_{NW}&=\R_{\u\u|\z}^{-1} \hspace{.1in} \makebox{and} \nonumber\\ \det[\R_{\u\u|\z}]&=\frac{1}{\det[(\R^{-1})_{NW}]}.
	\end{align} 
	This result was known to Harold Cramer more than 70 years ago \cite{Cramer:1946}, and is featured prominently in the book on graphical models by Whittaker \cite{Whittaker:2009}.  The consequence is that the ingredients of a partial coherence statistic we shall compute from experimental data may be read out of the  Northwest blocks of inverse sample covariance matrices.

%	under the  hypothesis $\makebox{H}_1: \R_{\x\y|\z}\neq {\bf 0}$, and $\text{blkdiag}\R_{\u\u|\z}$ 
%That is,
%\begin{equation}
%\R_{\u\u|\z}=\left[\begin{array}{cccc} \R_{\x\x|\z} & \R_{\x\y|\z}\\ \R_{\x\y|\z}^T & \R_{\y\y|\z} \end{array}\right] \makebox{ and } 
%{\color{ao}\text{blkdiag}\R}_{\u\u|\z}=\left[\begin{array}{cccc} \R_{\x\x|\z} & {\bf 0}\\ {\bf 0}^T & \R_{\y\y|\z} \end{array}\right]
%\end{equation}
The error covariance matrix $\R_{\u\u|\z}$ may be pre- and post-multiplied by %the square root of {\color{blue} the inverse of 
 $\text{blkdiag}[\R_{\x\x|\z}^{-1/2}, \R_{\y\y|\z}^{-1/2}]$  to produce the normalized error covariance matrix  and its corresponding determinant: 
\begin{align}\label{norm}
\Q_{\u\u|\z}
=&\left[\begin{array}{cccc} \I&\R_{\x\x\mid\z}^{-1/2}\R_{\x\y\mid\z}\R_{\y\y\mid\z}^{-1/2}\\\R_{\x\x\mid\z}^{-1/2} \R_{\x\y\mid\z}^T\R_{\y\y\mid\z}^{-1/2}&\I \end{array}\right],\\
\det[\Q_{\u\u|\z}]=&\frac{\det[{\R_{\u\u|\z}}]}{\det[\R_{\x\x|\z}]\det[\R_{\y\y|\z}]}=\det[\I-\C_{\x\y|\z}\C^T_{\x\y|\z}],\\
\C_{\x\y|\z}=&\R_{\x\x\mid\z}^{-1/2}\R_{\x\y\mid\z}\R_{\y\y\mid\z}^{-1/2}.
\end{align}
%The composite error covariance matrix of \eqref{error} may be pre- and post-multiplied by the block diagonal matrix $\makebox{blkdiag}[\R_{\x\x|\z}^{-1/2}, \R_{\y\y|\z}^{-1/2}]$ to produce the following normalized error covariance matrix: 
%\begin{align}\label{norm}
%\R_{\u\u|\z}^N%&\makebox{blkdiag}[\R_{\x\x|\z}^{-1/2}, \R_{\y\y|\z}^{-1/2}]\R_{\u\u|\z}\makebox{blkdiag}[\R_{\x\x|\z}^{-1/2}, \R_{\y\y|\z}^{-1/2}]\\=&
%=\left[\begin{array}{cccc} \I&\R_{\x\x\mid\z}^{-1/2}\P_{\x\y\mid\z}\R_{\y\y\mid\z}^{-1/2}\\\R_{\x\x\mid\z}^{-1/2} \P_{\x\y\mid\z}^T\R_{\y\y\mid\z}^{-1/2}&\I \end{array}\right]
%\end{align} 
 The NE matrix block of the covariance matrix ${\Q_{\u\u|\z}}$ is the {\em partial coherence matrix} %$\C_{\x\y|\z}$: % and its determinant, $\nu_{\x\y|\z}$ {\em partial coherence}: 
$\C_{\x\y|\z}$. %=\R_{\x\x\mid\z}^{-1/2}\R_{\x\y\mid\z}\R_{\y\y\mid\z}^{-1/2}$. 
The determinant of $\Q_{\u\u|\z}$ is  proportional to the volume of the normalized error concentration ellipse for estimating $\u$ from $\z$. When specialized for causal influence, this gives an  interpretation of the Granger-Geweke causality measure that is consistent with, but slightly different than,  the discussion given in \cite{Barnett1}. %, which was derived as a likelihood ratio in a MVN experiment. 
 Under the null hypothesis  $\makebox{H}_0$, the normalized error covariance matrix is identity, $\Q_{\u\u|\z}=\I_{p+q}$,  its determinant is one, and the partial coherence matrix is the zero matrix, $\C_{\x\y|\z}=\0_{p\times q}$.

We define  {\em partial coherence} to be 
\begin{align}
 \rho^2_{\x\y|\z}=&1-\det[{\Q_{\u\u|\z}}] 
=1-\frac{\det[\R_{\u\u|\z}]}{\det[\R_{\x\x|\z}]\det[\R_{\y\y|\z}]}\nonumber \\
=&1-\det[\I-\C_{\x\y|\z}\C^T_{\x\y|\z}].
\end{align}
Define the SVD of the partial coherence matrix to be ${\bf C}_{\x\y\mid\z}=\F\K\G^T$,
%\begin{align}\label{SVD}
%{\bf C}_{\x\y\mid\z}=\R_{\x\x\mid\z}^{-1/2}\P_{\x\y\mid\z}\R_{\y\y\mid\z}^{-1/2}=\F\K\G^T
%\end{align}
where $\F$ is a $p\times p$ orthogonal matrix, $\G$ is a $q\times q$ orthogonal matrix, and $\K$ is a $p\times q$ diagonal matrix of partial canonical correlations. The matrix $\K$ may be called the {\em partial canonical correlation matrix}.  The normalized error covariance matrix of eqn \eqref{norm} may be written as
\begin{align}
\Q_{\u\u|\z}=\left[ \begin{array}{cccc}\F & {\bf 0}\\{\bf 0} & \G  \end{array} \right]\left[ \begin{array}{cccc}\I & \K\\\K^T & \I  \end{array} \right]\left[ \begin{array}{cccc}\F^T& {\bf 0}\\{\bf 0} & \G^T  \end{array} \right].
\end{align}
As a consequence, partial coherence may be factored as %The determinant of this matrix measures the volume of the normalized error covariance matrix:
\begin{align}
 \rho^2_{\x\y|\z}%=\makebox{det}[\R_{\u\u|\z}^N]=&\frac{\makebox{det}[\R^1_{\u\u|\z}]}{\makebox{det}[\R_{\x\x|\z}]\makebox{det}[\R_{\y\y|\z}]}
=1-\makebox{det}[\I-\K\K^T]=1-\prod_{i=1}^{\makebox{min}(p,q)} (1-k_i^2). \label{cc}
\end{align}
%We shall call the normalized  error covariance $\R_{\u\u|\z}^N$ the {\em partial coherence matrix}, and its determinant {\em partial coherence}, denoted $\rho^2_{\x\y|\z}$. 
The  partial canonical correlations  %singular values 
$k_i$ are bounded between $0$ and $1$, as is partial coherence. %the determinant of the normalized error covarilance matrix. 
When the squared partial canonical correlations $k_i^2$ are near to zero, then partial coherence $\rho^2_{\x\y|\z}$ is near to  zero, indicating linear independence %in the MVN case conditional independence 
of $\x$ and $\y$, given $\z$. %If one prefers to see a measure of partial coherence near to one  for conditional independence, then partial coherence may be subtracted from one. This last sentence maybe dropped.

These results  summarize the error analysis for  linearly estimating the random vectors $\x$ and $\y$ from a common random vector~$\z$.  The only assumption is that the random vectors $(\x, \y, \z)$ are second-order random vectors. %The framework is general enough to accomodate arbitrary samples from the time series $\{x_t\}, \{y_t \}, \{z_t\}$.}  
By constructing the random vectors $\x,\y,\z$ appropriately, we shall answer questions of second-order causal influence.

\subsection{Regressing one random vector onto two: partial correlation and partial canonical correlation coefficients}
Suppose now that the random vector $\x$ is to be linearly regressed onto  $\v=(\y, \z)$: % $\z$ {\em and} $\y$. This estimator is easily read out of the composite covariance matrix as 
\begin{equation}
\hat{\x}( \v)=\left[ \begin{array} {cc} \R_{\x\y}&\R_{\x\z} \end{array} \right] \left[ \begin{array} {cccc} \R_{\y\y} & \R_{\y\z}\\\R_{\y\z} ^T& \R_{\z\z}\end{array} \right]^{-1}\left[ \begin{array} {cc} {\bf y}\\ {\bf z} \end{array} \right].
\end{equation}
Give the matrix inverse in this equation, the following block diagonal LDU factorization:
\begin{align}
&\left[ \begin{array} {cccc}  \R_{\y\y}&\R_{\y\z} \\ \R_{\y\z}^T & \R_{\z\z} \end{array} \right]^{-1} =\nonumber\\
&\left[ \begin{array} {cccc} \I &{\bf 0} \\ -\R_{\z\z}^{-1}\R_{\y\z}^T& \I \end{array} \right]\left[ \begin{array} {cccc} \R_{\y\y\mid\z}^{-1}&{\bf 0} \\  {\bf 0}& \R_{\z\z}^{-1} \end{array} \right]\left[ \begin{array} {cccc} \I &-\R_{\y\z}\R_{\z\z}^{-1} \\ {\bf 0}& \I \end{array} \right].
\end{align}
%where $\R_{\y\y\mid\z}=\R_{\y\y}-\R_{\y\z}\R_{\z\z}^{-1}\R_{\z\y}$ is the error covariance matrix when estimating ${\bf y}$ from ${\bf z}$. 
A few lines of algebra produce this result for $\hat{{\bf x}}( \v)$, the linear minimum mean-squared error estimator of $\x$ from $\v$:
\begin{align}
\hat{\x}( \v)=\hat{{\bf x}}({\bf z})+\R_{\x\y\mid\z} \R_{\y\y\mid\z}^{-1}\left[{\bf y}-\hat{{\y}}({\bf z})\right].
\end{align}
%In this formula $\hat{{\bf x}}({\bf z})$ is the minimum mean-squared error  estimator of ${\bf x}$ from ${\bf z}$ and 
%${\bf y}-\hat{{\bf y}}({\bf z})$ is the {\em error} of the minimum mean-squared error estimator of ${\bf y}$ from ${\bf z}$. Moreover $\R_{\y\y\mid\z}$ is the covariance matrix of this error. The matrix-valued  partial correlation matrix ${\bf S}_{\x\y\mid\z}=\R_{\x\y}-\R_{\x\z}\R_{\z\z}^{-1}\R_{\z\y}$ is the correlation between %${\bf x}$ and $({\bf y}-\hat{{\bf y}}({\bf z}))$, but it is also the {\em partial correlation} between  
%$({\bf x}-\hat{{\bf x}}({\bf z}))$ and
%$({\bf y}-\hat{{\bf y}}({\bf z}))$. %, by virtue of the fact that the error in estmating ${\bf y}$ from ${\bf z}$ is orthogonal to any linear function of ${\bf z}$, such as $\hat{{\bf x}}({\bf z})$.
%Let us summarize this result:
%\begin{align}
%\hat{{\bf x}}({\bf y},{\bf z})&=\hat{{\bf x}}({\bf y}-\hat{{\bf y}}({\bf z}))+\hat{{\bf x}}({\bf z})\\
%\hat{{\bf x}}({\bf z})&=\R_{\x\z}\R_{\z\z}^{-1}{{\bf z}}\\
%\hat{{\bf x}}({\bf y}-\hat{{\bf y}}({\bf z}))&=\S_{\x\y\mid\z} \R_{\y\y\mid\z}^{-1}({\bf y}-\hat{{\bf y}}({\bf z}))\\
%\hat{{\bf y}}({\bf z})&=\R_{\y\z}\R_{\z\z}^{-1}{\bf z}
%\end{align}
It is evident that the vector $\y$ is not used in a linear minimum mean-squared error estimator of $\x$ when the partial covariance $\R_{\x\y|\z}$  is zero. That is, the random vector $\y$ brings no useful second-order information to the problem of linearly estimating $\x$. %By constructing the random vectors $\x,\y,\z$ appropriately, we shall answer questions of second-order causality.

The error covariance matrix for estimating ${\bf x}$ from $ \v$ is easily shown to be 
\begin{align}
\R_{\x\x\mid\v}={\mathbb E}[(\x-\hat{\x}(\v)(\x-\hat{\x}(\v))^T  ] \nonumber\\
=\R_{\x\x\mid\z}-\R_{\x\y\mid\z} \R_{\y\y\mid\z}^{-1}  \R_{\x\y\mid\z}^T.
\end{align}
%where, to repeat,  $\R_{\x\x\mid\z}=\R_{\x\x}-\R_{\x\z}\R_{\z\z}^{-1}\R_{\z\x}$ is the error covariance matrix for estimating ${\bf x}$ from ${\bf z}$ only, $\R_{\y\y\mid\z} =\R_{\y\y}-\R_{\y\z}\R_{\z\z}^{-1}\R_{\z\x}$ is the error covariance for estimating ${\bf y}$ from ${\bf z}$ only, and $\S_{\x\y\mid\z}=\R_{\x\y}-\R_{\x\z}\R_{\z\z}^{-1}\R_{\z\y}$ is the partial correlation between ${\bf x}$ and ${\bf y}$ after regressing on a common ${\bf z}$. 
Thus the error covariance ${\R_{\x\x\mid\z}}$  is  reduced by a quadratic form depending on the  covariance between the errors  $({\bf x}-\hat{{\bf x}}({\bf z}))$ and  $({\bf y}-\hat{{\bf y}}({\bf z}))$. %The error covariance $\R_{\x\x\mid\y\z}$ is the error covariance in estimating $\x$ from $(\y,\z)$.  
If this error covariance is now normalized by the error covariance matrix achieved by regressing only on the vector $\z$, the result is 
\begin{align}
{\P_{\x\x|\v}}=&\R_{\x\x|\z}^{-1/2}\R_{\x\x|\v}\R_{\x\x|\z}^{-1/2}\nonumber\\=&\I-\R_{\x\x|\z}^{-1/2}\R_{\x\y\mid\z} \R_{\y\y\mid\z}^{-1} \R_{\x\y\mid\z}^T\R_{\x\x|\z}^{-1/2} \nonumber\\
=&\I-\C_{\x\y|\z}\C_{\x\y|\z}^T=\F(\I-\K\K^T)\F^T. 
\end{align}
The determinant of this matrix measures the volume of the normalized error covariance matrix:
\begin{align}
\makebox{det}[\P_{\x\x|\v}]&=\frac{\makebox{det}[\R_{\x\x|\v}]}{\makebox{det}[\R_{\x\x|\z}]%\makebox{det}[\R_{\y\y|\z}]
}\nonumber\\&=\makebox{det}[\I-\K\K^T]=\prod_{i=1}^{\makebox{min}(p,q)} (1-k_i^2) .
\end{align}
As before, we may define a {\em partial coherence} 
\begin{align}
 \rho^2_{\x|\y\z}=1-\det[\P_{\x\x|\v}]=1- \prod_{i=1}^{\makebox{min}(p,q)} (1-k_i^2). 
\end{align}
When the squared partial canonical correlations $k_i^2$ are near to zero, then partial coherence %dependence 
$\rho^2_{\x|\y\z}$ is near to  zero, indicating  linear independence %in the MVN case conditional independence 
of $\x$ on $\y$, given $\z$. Consequently  the estimator $\hat{\x}(\v)$ depends only on $\z$, and not on $\y$. %is algebraically independent of  $\y$. 

These results summarize the error analysis for estimating the random vector $\x$ from the composite vector $\v$.    It is notable that, except for scaling constants dependent only upon the dimensions $p,q,r$, 
the volume of the normalized error covariance matrix for estimating $\x$ from $\v$ equals the volume of the normalized error covariance matrix for estimating $\u$ from $\z$. %If these parameters are the same in both problems, then  $\det[\R_{\x\x\mid \v}^{v/z}]=\det[\R_{\u\u\mid \z}^{1/0}]$. %=\det[\I-\C_{\x\y|\z}\C^T_{\x\y|\z}]$. %The partial coherence  $\rho^2_{\x\y|\z}=\det[\R_{\u\u|\z}^N]=\det[\R_{\x\x|\v}^N]$ is the  coherence between $\x$ and $\y$ after regression onto $\z$. 
Both of these volumes are determined by the %partial correlation matrix {\color{blue} $\P_{\x\x|\v}$, or equivalently the partial coherence matrix $\C_{\u|\z}$} and its corresponding 
partial canonical correlations $k_i$.   Importantly, for answering questions of causal influence, %in Gaussian time series, 
based on linear prediction, it makes no difference whether one considers $\rho^2_{\x\y|\z}$ or $\rho^2_{\x|\y\z}$.  These two measures of coherence are identical. %They  are identical. 
So, for example, to analyze $\rho^2_{x_sy_t|\z}$ is to analyze $\rho^2_{x_s|y_t\z}$, for $\x$ defined to be $x_s$ and $\y$ defined to be $y_t$. In  Section~III.D %the following subsection 
we show that in the MVN case, the partial coherence between $\x$ and $\y$, given $\z$ is a monotone function of conditional KL divergence, conditional mutual information,  conditional information rate, and a %.  In the case of time series that are constructed to analyze causal influence,  it is a monotone function of the 
Granger-Geweke causality measure. 

Importantly the %squared 
partial canonical correlations $k_i$ are invariant to  transformation of the random vector $(\x^T, \y^T, \z^T)^T$ by a  block diagonal, non-singular, matrix $\T=\makebox{blkdiag}[\T_{\x\x}, \T_{\y\y}, \T_{\z\z}]$. As a consequence partial coherence is invariant to transformation $\T$. This result is the finite-dimensional version of Solo's finding \cite{Solo}  %for time series 
that the Granger-Geweke causality measure is invariant to uncoupled linear filtering of wide-sense stationary time series. A slight variation on Proposition 10.6 in \cite{Eaton} shows  partial canonical correlations  to be maximal invariants under group action $\T$.

%{\color{blue} \subsection{Likelihood}
%I am atrying to get all of this straight. $D_{KL}(P_0||P_1)=E_{P_0}\log \frac{P_0}{P_1}$. In the case of  a MVN. this is 
%\begin{equation}
%D_{KL}(P_0||P_1)=(1/2)\log\frac{\R^1_{\u\u|\z}}{\R^0_{\u\u|\z}}-L+\makebox{tr}[(\R^1_{\u\u|\z})^{-1}\R^0_{\u\u|\z}]
%\end{equation}

%Let's denote the conditional distribution of $\u$, given $\z$, as  $\u|\z \sim {\cal C}_{p+q}[{\bf 0},\R^0_{\u\u|\z}]$, and 
%consider the test of null hypothesis $H_0: \u|\z \sim {\cal C}_{p+q}[{\bf 0},\R^0_{\u\u|\z}]$ versus the alternative $H_1: \u|\z\sim {\cal C}_{p+q}[{\bf 0},\R^1_{\u\u|\z}]$. The likelihood ratio is 
%\begin{align}

%\end{align}}
%%%%%%%
%%%%%%%%%%%%%
%%%%%%%%%%%%%%%%%%%

\subsection{Spectral formulas}

There are spectral formulas. Begin with wide-sense stationary scalar time series $\{x_n\},\{y_n\},  \{z_n\}$, from which  error time series are computed. Assume the error covariance matrices $\R_{\x\x|\z}$, $\R_{\y\y|\z}$, and  $\R_{\x\y|\z}$ are Toeplitz matrices constructed from the scalar time series of error variances, $\{R_{xx|z}[n]\}$, $\{R_{yy|z}[n]\}$, and $\{R_{xy|z}[n]\}$.   Then the formulas of Szego \cite{Szego1915}, \cite{Bottcher1990} may be applied to show that in the limit $p,q\longrightarrow \infty$,
\begin{align}
\rho^2_{xy|z}%=&1-\frac{\det[\R^1_{\u\u|\z}]}{\det[\R^0_{\u\u|\z}]}
&\longrightarrow %\exp\{\int_{0}^{2\pi}\log\frac{\det[\hat{\R}^1(e^{j\theta})]}{\det[\hat{\R}^0(e^{j\theta})]}\frac{d\theta}{2\pi}\}
1- \exp\{\int_{0}^{2\pi}\log\left(1-\hat{k}^2(e^{j\theta})\right)\frac{d\theta}{2\pi}\},\\
\hat{k}^2(e^{j\theta})=&\frac{|\hat{R}_{xy|z}(e^{j\theta})|^2}{\hat{R}_{xx|z}(e^{j\theta})\hat{R}_{yy|z}(e^{j\theta})}.
\end{align}
where %$\hat{\R}^1(e^{j\theta})$ is the cross-spectral matrix of $\R^1_{\u\u|\z}$, and $\hat{\R}^0(e^{j\theta})$ is the cross-spectral matrix of $\R^0_{\u\u|\z}$;  
$\hat{R}_{xx|z}(e^{j\theta})$ is the  discrete-time  Fourier transform (DTFT) of  $\{R_{xx|z}[n]\}$, $\hat{R}_{yy|z}(e^{j\theta})$ is the DTFT of $\{R_{yy|z}[n]\}$, and $\hat{R}_{xy|z}(e^{j\theta})$ is the DTFT of  $\{R_{xy|z}[n]\}$. This formula generalizes to vector-valued time series \cite{Ramirez2010}. 
%These formulas show that the partial canonical coherence spectrum $\hat{k}(e^{j\theta})$  determines   conditional KL divergence, conditional information rate, and partial coherence. 
We may say the narrowband partial canonical correlations $\hat{k}(e^{j\theta})$ resolve the broadband partial coherence $\rho^2_{xy|z}$, in direct correspondence to eqn (\ref{cc}). Similar formulas may be found in \cite{Scharf2000}. %We may say the broadband {\color{red} partial coherence} %conditional information rate 
%decomposes as a sum of narrowband conditional information rates, where the narrowband rate at frequency $\theta$ is $-(1/2)\log(1-\hat{k}^2(e^{j\theta}))$ \cite{Scharf2000}.  
The partial canonical %coherence 
correlation spectrum $\{\hat{k}(e^{j\theta}), 0<\theta\le 2\pi\}$ is invariant to time-invariant  linear filtering of the wide-sense stationary time series $\{x_n\},\{y_n\},  \{z_n\}$. %, from which the error times series are computed. }%is the component of rate concentrated at frequency $\theta$.
 % This latter factors into the cross-spectral matrices of $\R_{\x\x|\z}$ and $\R_{\y\y|\z}$. 

%Assume the dimensions $p,q,r$ are equal. The essential object is the matrix $\R_{\u\u|\z}^N$. When each block-structured covariance matrix in the composite covariance matrix $\R$ is circulant, then each of the covariance blocks in $\R_{\u\u|\z}^N$ is circulant, and  KL divergence and information rate may be written {\color{blue} Where did the factor on 1/2 come from?}
% \begin{equation}
%R=\sum_{i=1}^nR(i); \makebox{ where } R(i)=\frac{1}{2}\log \frac{1}{1-k^2(i)}
%\end{equation}
%where $k^2(i)$ is the following ratio of discrete spectrum values:
%\begin{equation}
%k^2(i)=\frac{|P_{\x\y|\z}(i)|^2}{Q_{\x\x|\z}(i)Q_{\y\y|\z}(i)}
%\end{equation}
%Each of the terms in this expression is a narrowband spectral value. For example, $Q_{\x\x|\z}(i)$ is the $i$-th diagonal element of the diagonal matrix $\V^H\R_{\x\x|\z}\V$, with $\V$ the $p\times p$ DFT matrix.  
%We may say that the broadband conditional information rate decomposes as the sum of conditional narrowband information rates, each of which increases as squared canonical correlation increases to one. 

%%%%%%%%%%%%%
%%%%%%%%%%%%%%%%%%

%\subsection{Reading coherence out of the information matrix}

%{\color{blue} Yuan: You can see that I am breaking out this next section as a separate section.}
 %J. Whittaker, {\em Graphical Models in Applied Multivariate Statistics}, John Wiley and Sons, Chicester, UK, 1990}.
%%%%
%%%%%%%%%%%%

\subsection{ Kullback-Leibler Divergence, Conditional Mutual Information,  Conditional Information Rate,  Partial Coherence,
and Granger-Geweke Causality Measure}\label{sec: KLdiv} %Partial Canonical Correlations

%Let $P_1$ denote the probability law for the random vectors $(\U,\Z)$ in the case where the error covariance matrix  $\R_{\u\u|\z}$ is structured as in eqn 12, and let  $P_0$ denote the probability law for the random vectors $(\U,\Z)$ in the case where the partial correlation $\P_{\u|\z}$ block in $\R_{\u\u|\z}$ is $ {\bf 0}$. 

Begin with the random vectors $(\x,\y,\z)$. The Kullback-Leibler divergence between two probability laws $P_1$ and $P_0$ for these is defined to be
\begin{equation} 
D_{KL}(P_1||P_0)=%-E_{P_1}\log\frac{P_0(\W)}{P_1(\W)}=
{\mathbb E}_{P_1}\log\frac{p_1(\x,\y,\z)}{p_0(\x,\y,\z)},
\end{equation}
where $p_i(\x,\y,\z)$ is the probabiity density function under probability law $P_i$. We may think of the random variables $p_1(\x,\y,\z)$ and $p_0(\x,\y,\z)$ as pdf random variables or belief state random variables. Their ratio $p_1(\x,\y,\z)/p_0(\x,\y,\z)$ is a likelihood ratio random variable. So KL divergence may be considered the expected value of the log-likelihood ratio for testing the model $P_1$ against the model $P_0$, under the hypothesis that measurements are drawn from the probability law $P_1$. %It is sometimes useful to consider also the expectation under probability law $P_0$ in which case the difference  between $D_{KL}(P_1||P_0)$ and $D_{KL}(P_0||P_1)$ is termed J-Divergence,  the mean difference between log-likelihoods under hypotheses $H_1$ and $H_0$. {\color{blue} Get this right.}\\

When $(\x,\y,\z)$ is composed as $(\u,\z)$, then $p(\x,\y,\z)=p(\u|\z)p(\z)$ and the ratio is $p_1(\u|\z)/p_0(\u|\z)$, provided the marginal distribution of $\z$ is the same under probability laws $P_1$ and $P_0$. In the case where $(\u,\z)$ is multivariate normal (MVN) with mean ${\bf 0}$ and the covariance matrix  $\R$ of eqn (2), it is not hard to show that the Kullback-Leibler divergence is %For our purposes, the random vector $\W$ is the composed as $\W= (\U,\Z)$, . Call the model $P^1(\u,\z)$ the probability density function for a multivariate normal law with zero mean and covariance matrix $\R_1$ and $P_0(\u,\z)$ the probability density function for the  multivariate normal law with zero mean and covariance matrix $\R_0$. With $(\u,\z)$ replaced by the random vectors $(\U,\Z)$ these pdfs may be called pdf random variables or belief state random variables. The Kullback-Leibler Distance between the distributions $P_1$ and $P_0$ is then  
\begin{equation} 
D_{KL}(P_1||P_0)=-\frac{1}{2}\log\det[\Q_{\u\u|\z}] .%=-\frac{1}{2}\log\frac{\det[\Q^1_{\u\u|\z}]}{\det[\Q^0_{\u\u|\z} ]}
\end{equation}
%For our purposes the covariances $\R_1$ is the covariance matrix for the random vectors $(\U,\Z)$, and $\R_0$ is the same covariance matrix in the case where the partial correlation coefficient is ${\bf 0}$. That is, $[\R_0^{-1/2}\R_1\R_0^{-1/2}=\R^N_{\u\u|\z}$.
That is, after regression of $\x$ and $\y$ onto %the third channel 
$\z$, the problem is reset to the analysis of the errors  $(\x-\hat{\x}(\z))$ and $(\y-\hat{\y}(\z))$, and their partial coherence. %with partial coherence matrix  $\R_{\u\u|\z}^N$. %If these errors are multivariate normal, then this this covariance is a complete second-order characterization of the probability law for the errors. Call the model $P_1$ the model that $\u\sim {\cal N}_{p+q}{\bf 0}, \R_{\u\u|\z}]$, with $\R_{\u\u|\z}$ structured as above, and $P_0$ the same probability law, but with $\R_{\u\u|\z}$ block-diagonally structured. That is the partial correlation matrix $\P_{\x\y|\z}={\bf 0}$. It is easy to show that the KL distance between these two probability laws is
%\begin{equation} 
%D_{KL}(P_1||P_0)=-E_{P_1}\log\frac{P_1}{P_0}=-\log\det[\R_{\u\u|z}^N]
%\end{equation}
%This may be considered the expected value of the log-likelihood ratio for testing the model $P_1$ against the model $P_0$. 
In the MVN case, KL divergence for this problem is also the conditional mutual information,  $I_{\x\y|\z}={\mathbb E}\left[ \log\frac{p(\x,\y|\z)}{p(\x|\z)p(\y|\z)} \right]$ between the random vectors $\x$ and $\y$, given $\z$. 
%\begin{equation}
%I_{\x\y|\z}=E\left[ \log\frac{p(\x,\y|\z)}{p(\x|\z)p(\y|\z)} \right]=\log\frac{\det[\R^0_{\u\u|\z}]}{\det[\R^1_{\u\u|\z}]}=D_{KL}(P_1||P_0)
%\end{equation}}
Thus KL divergence is also Shannon's definition of the rate $R$ at which the MVN error $\y-\hat{\y}(\z)$ brings information about the MVN error $\x-\hat{\x}(\z)$. This is also the rate at which $\y$ brings information about $\x$, given $\z$, and a monotone function of what the Granger-Geweke causality measure would say for finite-dimensional random vectors $(\x,\y, \z)$. %,   with appropriate definitions of $(\x,\y,\z)$.
So partial coherence, conditional KL divergence, conditional mutual information, conditional  information rate, % partial coherence, 
and a Granger-Geweke causality measure % adapted to finite-dimensional random vectors are  %(volume of normalized error ellipse) 
are  monotone functions of each other. All are resolved by the partial canonical correlations $k_i$:
%\begin{equation}
%\det[\R_{\u\u|\z}^N]=e^{-D_{KL}}=e^{-R}=\det[\I-\K\K^T]=1-\prod_{i=1}^p(1-k_i^2)
%\end{equation}
\begin{align}
D_{KL}(P_1||P_0)=&I_{\x\y|\z}%=R
=-\frac{1}{2}\log (1- \rho^2_{\x\y|\z})\\
=& -\frac{1}{2}\sum_{i=1}^{\makebox{min}(p,q)}\log(1-k_i^2) %R_i \makebox{ where } R_i=-\log (1-k_i^2).
\end{align}
 These finite-dimensional connections are consistent with the finding in \cite{Barnett2} that the Granger-Geweke causality measure and transfer entropy are equivalent for MVN variables. %To repeat, 
These equivalent measures are invariant to block diagonal nonsingular transformations of the random vectors $\x, \y, \z$, and they are functions of the  partial canonical correlations $k_i$, which are maximal invariants.

When the partial canonical correlations $k_i$ are near to zero, indicating the coherence matrix $\C_{\x\y|\z}$ is near to the zero matrix, then partial coherence is near to  zero, %one, 
and  conditional KL divergence, conditional mutual information,  and conditional rate are near to zero. The Granger-Geweke causality measure is near to one, and its transfer entropy is near to zero.\\

%When the partial canonical correlations $k_i$ are near to one, then  conditional KL divergence, conditional mutual information, and condtional rate are large, and vice-versa. %When KL divergence is large, rate is large, and partial coherence is small, %the volume of the error covariance matrix $\R_{\u\u|\z}^N$ is small 
%and vice-versa. The partial canonical correlations $k_i$ provide an additive decomposition of rate.

%{\color{blue} Yuan: You can see that I fixed the dependence on the $k_i^2$ in the last equation above.}

\section{ Partial Coherence, Likelihood, and a Test for Causality}\label{sec:test}

Partial coherence is inherently connected to likelihood in the case where the random vectors $(\x,\y,\z)$ are jointly normally distributed. The matrix $\R_{\u\u|\z}$ %denote 
is the conditional error covariance matrix for the error $(\u-\hat{\u}(\z))$; when    the partial correlation $\R_{\x\y|\z}$ is zero, %non-zero, and 
then $\R_{\u\u|\z}=\makebox{blkdiag}[\R_{\x\x|\z}, \R_{\y\y|\z}]$. % R_{\u\u|\z}^0$

%$\R_{\u\u|\z}^0$ the error covariance matrix in the case $\P_{\x\y|\z} ={\bf 0}$. 
The null hypothesis test of interest is $\makebox{H}_0: (\u-\hat{\u}(\z))\sim \mathcal{N}_{p+q}({\bf 0},\makebox{blkdiag}[\R_{\x\x|\z}, \R_{\y\y|\z}])$ versus the alternative
$\makebox{H}_1: (\u-\hat{\u}(\z))\sim \mathcal{N}_{p+q}({\bf 0},\R_{\u\u|\z})$. %By the conditional independence of the error $(\u-\hat{\u}(\z))$ and $\z$, this is equivalent to a composite binary  hypotheses test on the pair $(\u,\z)$.}

%, or in fact uniformly distributed with respect to left orthogonal transformations.  {\color{red} Louis: Check this. Get in comments earlier about the distribution of u given z and the corresponding error covariance matrix Q being invariant to block diagonal transformation.} 

\subsection{Likelihood}
Given the data matrix $\D^T=(\X^T,\Y^T, \Z^T)=(\U^T,\Z^T)$, where each of the matrices $\X\in{\mathbb R}^{p\times M},\Y\in{\mathbb R}^{q\times M}, \Z\in{\mathbb R}^{r\times M}$ consists of $M$ i.i.d. realizations of the jointly multivariate normal random vectors $\x,\y, \z$, compute the sample covariance matrix $\S=\D\D^T$, structured as follows:
\begin{align}
\S&=\left[\begin{array}{cccc}\S_{\x\x}&\S_{\x\y}&\S_{x\z}\\ \S_{\y\x}&\S_{\y\y}&\S_{y\z}\\\S_{\z x}&\S_{\z y}&\S_{\z\z}\end{array} \right]=\left[\begin{array}{cccc} \S_{\u\u}&\S_{\u\z}\\ \S_{\z\u}&\S_{\z\z}\end{array} \right].
\end{align} 
For example, $\S_{\x\y}=\X\Y^T$. 

The ordinary likelihood for testing $\makebox{H}_0$ versus $\makebox{H}_1$ is then \cite{Mardia}
 \begin{equation}
1-\hat{\rho}^2_{\x\y|\z}%= \frac{\det[\S_{\u\u|\z}^1]}{\det[\S_{\u\u|\z}^0]}
= \frac{\det[\S_{\u\u|\z}]}{\det[\S_{\x\x|\z}]\det[\S_{\y\y|\z}]},%=\det[\S_{\u\u|\z}^N] 
\end{equation}
where 
\begin{align}
 \S_{\u\u|\z}=\S_{\u\u}-\S_{\u\z}\S_{\z\z}^{-1}\S_{\u\z}^T, \\
  \S_{\x\x|\z}=\S_{\x\x}-\S_{\x\z}\S_{\z\z}^{-1}\S_{\x\z}^T,\\
   \S_{\y\y|\z}=\S_{\y\y}-\S_{\y\z}\S_{\z\z}^{-1}\S_{\y\z}^T.
 \end{align}
Of course, as noted previously, the matrix $\S_{\u\u|\z}^{-1}$ may be read out of the inverse $\S^{-1}$ and used to compute the determinant $\det[\S_{\u\u|\z}]$. 

The following proposition establishes the distribution, and a stochastic representation, for the likelihood ratio statistic $\hat{\rho}^2_{\x\y|\z}$ under the null hypothesis. This result forms the basis for setting a detection threshold to achieve a desired type-I error. %of a stochastic simulation, from which a threshold may be computed for control of type-I errors.\\

\begin{Prop}
Under the null hypothesis $\makebox{H}_0$ that the coherence matrix %partial coherence $\rho^2_{\x\y|\z}$ 
$\C_{\x\y|\z}$ is the zero matrix, the %partial coherence 
likelihood ratio statistic  $1- \widehat{\rho}^2_{\x\y|\z}$ has the Wilks $\Lambda$ distribution, $\Lambda(p, M-r-q, q)$ with stochastic representation 
$$
\widehat{\rho}^2_{\x\y|\z}\sim 1-\prod_{i=1}^p b_i,
$$
where the $b_i$ are independent Beta random variables with $\makebox{Beta}(\frac{M-r-q-i+1}{2}, \frac{q}{2})$ distributions.
\end{Prop}
 
\begin{proof} 
 Let $\U_1=\X-\X\Z^T(\Z\Z^T)^{-1}\Z$, $\U_2=\Y-\Y\Z^T(\Z\Z^T)^{-1}\Z$, and $\U=[\U_1^T, \U_2^T]^T$. It follows that %One can show that 
$\U\U^T= \S_{\u\u|\z}$, $\U_1\U_1^T= \S_{\x\x|\z}$, 
$\U_2\U_2^T= \S_{\y\y|\z}$,  and therefore,  
 \begin{align}
 1-\widehat{\rho}^2_{\x\y|\z}=\frac{\det(\U\U^T)}{\det(\U_1\U_1^T)\det(\U_2\U_2^T)}.
 \end{align}
 Kshirsagar \cite{Kshirsagar} has shown that the error covariance $\U\U^T$ has a Wishart distribution $W_{p+q}(\R_{\u\u|\z}, M-r)$ where $r$ is the dimension of $\z$. The statistic $\frac{\det(\U\U^T)}{\det(\U_1\U_1^T)\det(\U_2\U_2^T)}$ is 
invariant to  left multiplication of $(\X,\Y)$ by a block-diagonal matrix $\left[\begin{array}{ccc}\T_1&0\\ 0&\T_2\end{array}\right]$ for non-singular 
$p\times p$  matrix $\T_1$ and $q\times q$ matrix $\T_2$. Choose $\T_1=\R_{\x\x|\z}^{-1/2}$ and $\T_2=\R_{\y\y|\z}^{-1/2}$, so that WLOG $\U\sim \mathcal{N}({\bm 0}, \I_{p+q}\otimes \I_{M-r} )$ which yields $\U_1\sim \mathcal{N}({\bm 0}, \I_{p}\otimes \I_{M-r} )$ and $\U_2\sim \mathcal{N}({\bm 0}, \I_{q}\otimes \I_{M-r} )$. Moreover,  $\U_1$ and $\U_2$ are independent under the null hypothesis

Let $[\V, \V^\perp]\in {\mathbb R}^{(M-r)\times (M-r)}$ be the orthogonal matrix where $\V\in  {\mathbb R}^{(M-r)\times q}$ is a unitary basis for the row space of $\U_2$ and $\V^\perp\in {\mathbb R}^{(M-r)\times(M-r-q)}$ is the unitary basis for its orthogonal space.  We can choose  $\V = \U_2^T(\U_2\U_2^T)^{-1}\U_2$ and define $\W_1=\U_1\V$ and $\W_2=\U_1\V^\perp$. Conditional on $\V$ and $\V^\perp$, it can be shown that $\W_1\sim \mathcal{N}(0, \I_p \otimes \I_q)$, and $\W_2\sim \mathcal{N}(0, \I_p \otimes \I_{M-r-q})$ (assuming $M-r>p+q$) which are invariant to $\V$.
Simple linear algebra leads to 
 \begin{align*}
1-\widehat{\rho}^2_{\x\y|\z}=\frac{\det(\W_2\W_2^T)}{\det(\W_1\W_1^T+\W_2\W_2^T)}.
\end{align*}
 %\begin{align*}
%\widehat{\rho}^2_{\x\y|\z}%&=\frac{\det(\tilde\U_2\tilde\U_2^T)\det(\tilde\U_1\tilde\U_1^T-\tilde\U_1\tilde\U_2^T(\tilde\U_2\tilde\U_2^T)^{-1}\tilde\U_2\tilde\U_1^T)}{\det(\tilde\U_1\tilde\U_1^T)\det(\tilde\U_2\tilde\U_2^T)}\\
%&=\frac{\det(\tilde\U_1\P_2^{\perp}(\P_2^{\perp})^T\tilde\U_1^T)}{\det(\tilde\U_1\P_2^{\perp}(\P_2^{\perp})^T\tilde\U_1^T+\tilde\U_1\P_2\P_2^T\tilde\U_1^T)}
%\end{align*}
This has the Wilk's $\Lambda$ distribution $\Lambda(p, M-r-q, q)$ \cite{Wilks}, which is identical to the product of $p$ independent beta-distributed random variables:
\begin{align*}
1-\widehat{\rho}^2_{\x\y|\z}\sim \prod_{i=1}^p b_i \text{ where } b_i \sim \makebox{Beta}(\frac{M-r-q-i+1}{2},\frac{q}{2}).
\end{align*}

\end{proof}
It is important to note that this distribution is invariant to the actual error covariances $\R_{\x\x|\z}$ and $\R_{\y\y|\z}$, which are themselves functions of the underlying composite covariance matrix $\R$ of eqn (2).  Moreover, this distribution result specializes to the case $r=0$, in which case partial coherence is ordinary multivariate-multivariate coherence, and the distribution of the  coherence statistic is   $1-\Lambda(p, M-q, q)$. When $p=1$, then partial coherence is multivariate coherence, with distribution $1-\Lambda(1, M-q, q)$.  
%Computations of the Wilks' distribution for higher dimensions can be time consuming.  %\cite{Bartlett1954} derived the asymptotic distribution that is as $M$ goes to infinity, the limiting distribution is $\mathcal{X}^2-$distribution with a chosen factor as
%$$
%-(M-r-\frac{p+q+1}{2})\log(1-\widehat{\rho}^2_{\x\y|\z})\sim \mathcal{X}^2_{pq}.
%$$
% }
 In \cite{Bartlett1954} Bartlett derived the %following 
asymptotic distribution for large $M$,
$$
-(M-r-\frac{p+q+1}{2})\log(1-\widehat{\rho}^2_{\x\y|\z})\sim \mathcal{X}^2_{pq},
$$
where $\chi^2_{pq}$ denotes a chi-squared sdistribution with $pq$ degrees of freedom.

%{\color{blue} These results hold for $(\x,\y,\z)$ jointly distributed as spherically-invariant random vectors.}\\
%{\color{blue} Yuan: Do we want to say anything about the asymptotic chi-squared distribution of likelihood, and reference S.S. Wilks, ``The large sample distribution of the likelihood ratio for testing composite hypotheses," Ann Math Stat, vol 6, no 1, pp 60-62 (1938)? }
%{\color{blue} Yuan: This result comes from variable counting and some other magic, which you might want to elaborate on. Besides Bartlett, you might want to see the agumentation in  Klausner, et al., "Detection of spatially correlated time series from a network of sensor arrays, 2014," and "Saddlepoint approximations for correlation testing among multiple Gaussian random vectors, 2016," both in IEEE Trans Signal Processing. These don't need to be referenced, but you might see something useful for you in them. And wouldn't it be nice to show that asymtotically the non-centrality parameter of the chi-squared is $1-\rho^2_{\x\y|\z}$?}

%\section{Causal Influence and Partial Coherence in Time Series}\label{sec:ts}
%{\color{red} Yuan: Now I am thinking your placement of this section here might be a good idea.}

\section{Experiments}\label{sec:exp}

 If partial coherence is to be used to answer questions of causal influence between signal sequences,  %time or space series, 
then there is required a princpled way to map a causality question into a question of conditional dependence among random vectors. %To this end, we begin with a narrative on predictive causality, and establish a mechanism for re-phrasing a question of predictive causality as a question of partial coherence.}

Our approach to causality analysis of time series will be to start with two channels of signals, %measurements, 
call them $\{x_n\}$ and $\{y_n \}$, and then to define a third signal channel $\{z_n\}$ of prima facie evidence in such a way that  the question of causality %between elements of these two channels 
may be resolved from an analysis of partial coherence. For example the %elements might be  
 random vectors $(\x,\y,\z)$ might be $\x=x_s$ and $\y=y_t$,  which are readings of the signal sequences %time series 
at the two times $s$ and $t$. 
%Let  $s$ and $t$ be two time points, we  define the random variables $\y=\y[s]$ and $\x=\x[t]$. 
In some applications, the prima facie evidence $\z$ might be defined as the past of the time series $\{y_n\}$ at time $t$, i.e.,  ${\bf z}=\{y_n, n< t\}$. In other applications the prima facie evidence might be  ${\bf z}=\{x_n, n< t\}$. As a practical matter in testing from experimental data, this past will be finite. More commonly, $\x$ is assembled from the finite past of $\{x_n\}$ up to time $t-1$, $\y$ is defined to be $y_t$,  and $\z$ is assembled from  the finite past of $\{y_n\}$ up to time $t-1$.  With $\x, \y, \z$ so defined, we can compute the partial coherence statistic $\rho_{\z}^2(s,t)\doteq \rho^2_{x_sy_t|\z},$ or more generally $\rho_{y_t\x|\z}^2$. %as introduced in Section~2.  
The statistic $\rho_{\z}^2(s,t)$ vs the pair $(s,t)$ reveals  pairwise causality %or non-causality 
with respect to prima facie evidence $\z$. The statistic $\rho_{y_t\x|\z}^2$ reveals the causal influence of the finite past of $\{x_n\}$ on $y_t$, conditioned on the finite past of $\{z_n\}$.

A value of partial coherence below a threshold is taken to be a finding of non-causality at significance level $\alpha$, whereas a value of coherence above this threshold is taken only to be an indication of causality with respect to this specific prima facie evidence and a second-order analysis of partial coherence. New evidence and/or new methods of analysis may invalidate an indication of causality. In the multivariate normal case, no new evidence or method of analysis can invalidate the finding  of non-causality, as we have found prima facie evidence and a method of analysis that indicates non-causality at significance level $\alpha$. An equivalent  view is that the likelihood ratio test is a null hypothesis test for non-causality. %, with the likelihood of measurements under the null hypothesis normalized by likelihood for the most likely model of causality. %Then an estimate of partial coherence falling below a decision threshold is taken as a finding of non-causality, but an estimate exceeding a decision threshold is taken only as a finding that, for the given prima facie evidence and method of analysis, causality is indicated.

To demonstrate the use of partial coherence for resolving the question of causality, we present in this section a demonstration and an experiment.  The demonstration is an analytical calculation of partial coherence for a known bi-variate time series which may model past-causal, future-causal, or mixed-causal influence of one time series upon the other, depending on the choice of coupling parameters between the two time series. The experiment is a re-running of the Barnett and Seth experiment \cite{Barnett1}. 

\subsection{ Demonstration of Partial Coherence for  Past-causal, Future-causal, and Mixed-Causal Models}
%To set the stage for the random experiment to be conducted in the next section, we illustrate the computation of partial coherence for two channels with causal, anti-causal, and mixed causality. %In this section, we will conduct several experiment and use the partial coherence to assess the causality between two time series.

Consider the following  multivariate  linear system, %as shown in the top left panel of Fig.~\ref{fig:exp1}, %the top left panel,  
 
	\begin{align*}
\x_n&=\sum_{k=-\infty}^{+\infty} \BH_k{\bm \mu}_{n-k}\\
\y_n&=\sum_{k=-\infty}^{+\infty} \BG_k {\bm \nu}_{n-k}+\sum_{k=-\infty}^{+\infty} \BF_k\x_{n-k},
\end{align*}
where  ${\bm \mu}, {\bm \nu}$ are  multivariate unit variance  white noises.  The two time series are correlated with auto- and cross-covariance matrices 
\begin{align*}
\R_{\x\x}[m]&={\mathbb E}[\x_n\x_{n+m}]=\sum_k \BH_k\BH_{k+m},\\
\R_{\y\y}[m]&={\mathbb E}[\y_n\y_{n+m}]\nonumber\\
&=\sum_k \BG_k\BG_{k+m}+\sum_{k,l}\BF_k\R_{\x\x}[m+k-l]\BF_l,\\ \nonumber
\R_{\x\y}[m]&={\mathbb E}[\x_n\y_{n+m}]=\sum_k \BF_k\R_{\x\x}[m-k]. \nonumber
\end{align*}
%In these formulas, the dimensions of $\x$ and $\y$ may be taken to be $p$ and $q$, and the dimensions of $\u$ and $\v$ may be taken to be $k$ and $\ell$. 
When these infinite-dimensional moving average (MA) models arise as finitely-parameterized autoregressive-moving average (ARMA) models, then these infinite sums may be evaluated analytically.  %As a result, for so-called rational models for $\F, \G$, and $\H$ the partial coherence can be  calculated analytically. % as introduced in Section 2.1. 

 For our demonstration of partial coherence, we assume the time series are scalar-valued. We set the filter coefficients %for the noises 
at $h_k=h_0a^k$ and $g_k=g_0b^k$ for $k \geq 0$ where $h_0=0.8$, $a=0.1$, $g_0=0.7$, and $b=0.7$. The coefficients  $\{f_k\}$, are the filtering coefficients that determine the causality  between $\{y_n\}$ and $\{x_n\}$.  For any selected pair of times $s,t$, we evaluate the partial coherence between $x_s$ and $y_t$.   We consider two sets of prima facie evidence:  $\z$ is the past of $\{y_n\}$ up to time  $t-1$, and $\z$ is the past of  $\{x_n\}$ up to time $t$ with $x_s$ excluded. 
Partial coherence %statistic 
is therefore
%\begin{equation}
$\rho^2_{\z}(s,t)\doteq \rho^2_{y_t x_s|\z}$.
%\end{equation}

We consider three different scenarios.  Case I: % the past causal dependency: 
conditioned on the finite past of the input $x$, the output $y$ is linearly independent %with 
of any input outside this finite past which includes the future of $x$; Case II:  %the future causal dependency: 
conditioned on the finite future of the input $x$ %finite past and future input $x$, 
the output $y$ is linearly independent of any input outside this finite future which includes the past of $x$; Case III; conditioned the finite past and finite future of the input $x$, the output $y$ is independent of any inputs outside this finite past and future.

\vspace{0.1in}
\noindent {\bf Case I}:   $\{y_n\}$ is past-causally dependent on $\{x_n\}$: 
\begin{align*}
y_n=& .7\sum_{k=0}^{+\infty} (.7)^k \nu_{n-k}+ .7x_n+.8x_{n-1}+ .7x_{n-2}+ .6x_{n-3}\\
x_n=&.8\sum_{k=0}^{+\infty} (.1)^k \mu_{n-k}.
\end{align*}
\noindent {\bf Case II}:  $\{y_n\}$ is future-causally dependent on $\{x_n\}$:
  \begin{align*}
   y_n=& .7\sum_{k=0}^{+\infty} (.7)^k \nu_{n-k}+.7x_n+.8x_{n+1}+ .7x_{n+2}+ .3x_{n+3}\\
x_n=&.8\sum_{k=0}^{+\infty} (.1)^k \mu_{n-k}.
   \end{align*}   
\noindent {\bf Case III}: the mixed case:
  \begin{align*}
 y_n=& .7\sum_{k=0}^{+\infty} (.7)^k \nu_{n-k}+.4x_{n-2}+.8x_{n-1}\\
 &\quad +.7x_{n}+ .8x_{n+1}+.4x_{n+2}\\
x_n=&.8\sum_{k=0}^{+\infty} (.1)^k \mu_{n-k}.
   \end{align*} 

 \begin{figure}[htbp] 
\centering
\includegraphics[width=0.8\linewidth]{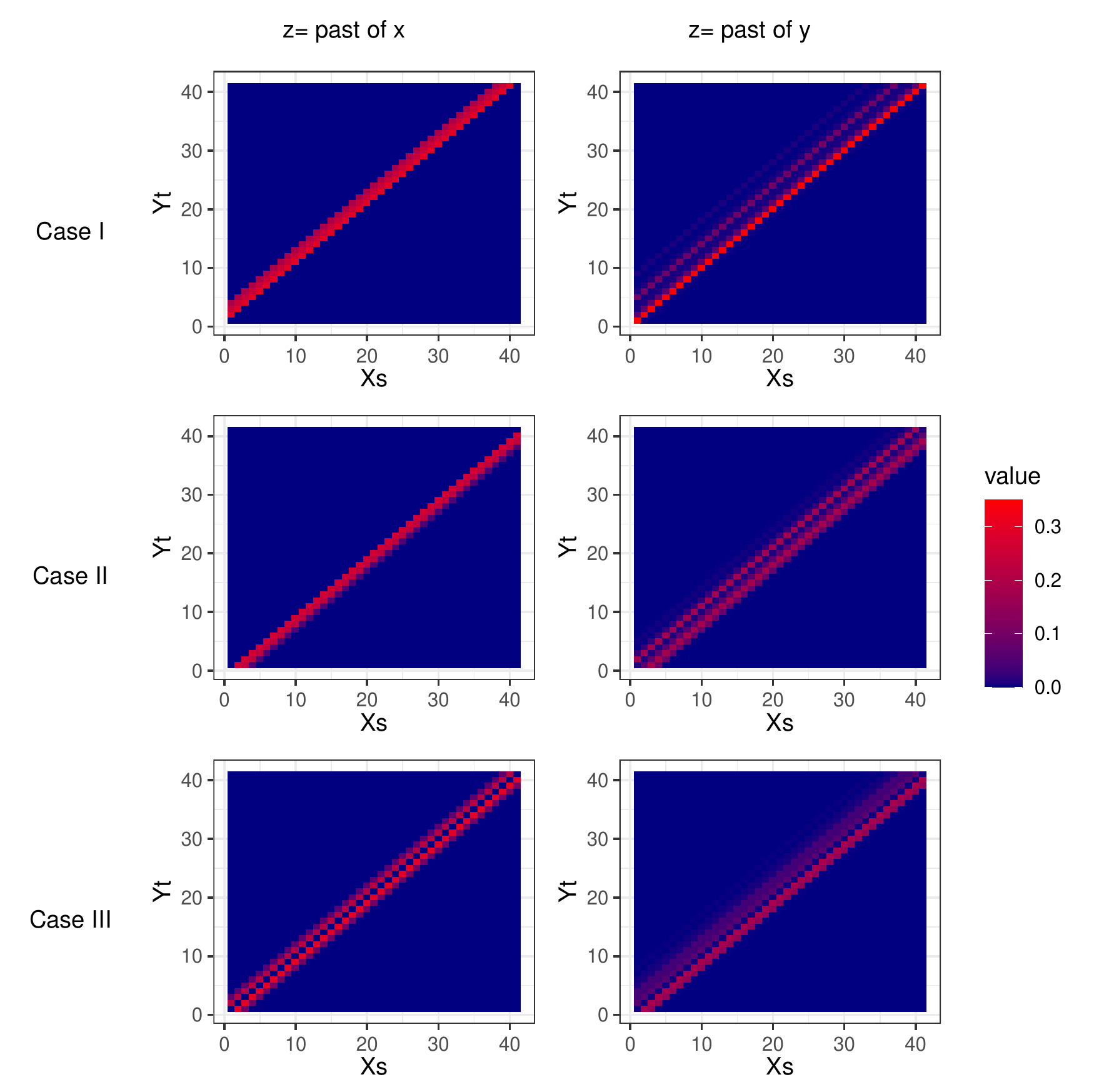}
\caption{Pairwise partial  coherence between $x_s$ at time $s$ (horizontal axis) and $y_t$ at  time $t$ (vertical axis). The prima facie evidence $\z$ is the past of $\{x_n\}$ up to time $t$ in the left column,  and the past of $\{y_n\}$ up to time $t-1$ in the right column; (top row) $\{y_n\}$ is past-causally dependent on $\{x_n\}$; (middle row) %(bottom left) 
$\{y_n\}$ is future-causally dependent on $\{x_n\}$;  (bottom row) $\{y_n\}$ has mixed past- and future-causal  dependence on $\{x_n\}$.} 
\label{fig:exp1}
\end{figure}

The resulting pairwise partial coherence maps are illustrated in Figure~\ref{fig:exp1}.  The three rows show the results of case I, II, and III, respectively. Each panel of Figure~\ref{fig:exp1} illustrates the partial coherence between $y_t$ and $x_s$, with the prima facie evidence being the past of $\{x_n\}$ on the left column, and the past of $\{y_n\}$ on the right column. As we already know,  the partial coherence are conditioned on available prima facie evidence. With $\z=\{x_t, x_{t-1}, \ldots\}$, (Notice that $x_s$ is excluded from $\z$ if $s\leq t$.),   $\rho^2_{\z}(s, t)=0$ for any $s>t$ or $s<t-3$ in Case I (top row),   which suggests that future and long-past values of $\{x_n\}$  do not contain any second-order information about $y_t$ given this $\z$ (In fact, they are conditionally independent in the MVN case). In Case II the future-causal case (middle row), $\rho^2_{\z}(s, t)=0$ for any $s\leq t$ since in this case the past of $\{x_n\}$ does not contain unique information about $y_t$ given $\z$.  In Case III (bottom row)  the mixed case where $y_t$ depends on both the past and future of $x$, $\rho^2_{\z}(s, t)>0$ for both $s\leq t$ and $s>t$. In summary, there is an asymmetry in the past-causal and future-causal cases that bears comment.  In Case I, $\rho^2_{\z}(s, t)$ for $s<t$ is non-zero, but weak, compared with $\rho^2_{\z}(s, t)$ for $s>t$ in Case II. This effect is explained by noting that for these demonstrations our conditioning is always on the past of $\{x_n\}$ at time $t$.  The right column of Figure~\ref{fig:exp1} shows the partial coherence between $y_t$ and $x_s$ given the past of $\{y_n\}$, i.e., $\z=\{y_{t-1}, y_{t-2},\ldots\}$ for the three cases. Here the partial correlation generally is nonzero for both $s<t$ and $s>t$. The reason is that $\{x_n\}$ is an auto-regressive process and $\{y_n\}$ is a finite moving-average of $\{x_n\}$. When the prima facie is the past of $\{y_n\}$, $x_s$ still contains some unique information about $y_t$ due to the
infinite memory of $\{x_n\}$.  While all three cases are Granger causal, the choice of prima facie evidence
as the past of input time series $\{x_n\}$ leads to clear separation between past-causal, future-causal, and the mixed-causal influence, whereas the conventional choice of prima facie evidence as the past of the output time series $\{y_n\}$ 
does not resolve the direction of causal influence.

%\begin{figure} 
%\centering
%\includegraphics[width=1\linewidth]{exp.png}
%\end{figure}

%When the random variables $x$ and $y$ are vector-valued, then the vector $\u$ is $\u=(\x^T, \y^T)^T$, and the error covariance matrix $\R_{\u\u}$ is 
%\begin{align*}
%{\bf Q}_{\ub\ub\mid \z}&=E\left[(\ub-\hat{\ub})(\ub-\hat{\ub})^H\right]\\
%&=\left[\begin{array}{cccc}\R_{xx\mid\z}&\S_{xy\mid\z}\\
%\S_{xy\mid\z}^H&\R_{yy\mid\z}  \end{array}\right].
%\end{align*}
%The NE term   $\R_{xx\mid\z}^{-1/2} \S_{xy\mid\z}\R_{yy\mid\z}^{-1/2}$ is the partial correlation matrix between $\x$ and $\y$ given $\z$. The determinant of the partial coherence matrix ${\bf C}_{xy\mid\z}$ is a Hadamard ratio, and it determines the coherence statistic we propose for testing the hypothesis that $\S_{xy\mid \z}$ is zero:
%\begin{equation*}
%\rho^2=1-\det[{\bf C}_{xy\mid\z}]=\frac{\det[ \S_{xy\mid\z}\S_{\y\x}]}{\det[\R_{xx\mid\z}]\det[Q_{yy\mid\z}]}.
%\end{equation*}

%This contrasts with Geweke's advocacy for the statistic ... . 

%Following the lines of Geweke\cite{Geweke82}, \cite{Gewele84}, and Barret \cite{Barnett1}, \cite{Barnett2}, This coherence statistic may be used to assess partial correlation, or linear dependence in vector-valued state space models. 
%{\color{red}Yuan: I have just written what I think we might want to say. But I am not at all certain that our coherence is different from Geweke's... see eqn (9) in Barnett1.}

\subsection{The experiment of Barnett, et al.}

%{\yw Previous work in \cite{Geweke82}-\cite{Barnett3} have studied the time-domain Granger causality among multiple time series. The causality from a time series $\{y_t\}$ to  $\{x_t\}$ is defined by the predictive power gained by adding the past of $\{y_t\}$ in predicting $x_t$. More specifically, the causality is measured by the logarithm ratio of determinant  for  $\Sigma_{xx}$ over  $\Sigma_{xx}^-$ where  $\Sigma_{xx}$ is the error covariance for predicting $\x_t$ using the past of $\x_t$, a third time series $\z_t$ if available, and $\y_t$, and  $\Sigma_{xx}^{-}$ is with past of $\y_t$ excluded. In the null case when $\{y_t\}$ has no causal impact on $\{x_t\}$, the log ratio is zero. 
 In this subsection, we replicate the experiment in Barnett et al. \cite{Barnett1} and use the partial coherence statistic to test for time-domain causality.   As in \cite{Barnett1}, we define $\y=y_t$, $\z=[y_{t-1}, y_{t-2}, \ldots]$, and $\x=[x_{t-1}, x_{t-2}, \ldots]$. The causality from time series $\{x_n\}$ to  $\{y_n\}$ is measured by the partial coherence between $\x$ and $\y$ given $\z$.  Straightforward calculation shows that %our 
partial coherence $\rho^2_{\x\y|\z}$ is a monotone function of the causality measure in \cite{Barnett1}.  However, as a {\em statistic}, our %the 
estimation of the causality measure is very different. The %evaluation of the 
test statistic in Barnett et al. is computed %done by estimating  %the 
by replacing state-space parameters by their approximate maximum likelihood estimates, %state-space model parameters, 
and inserting these estimates into an analytical formula for their ratio of determinants.  The result is an approximate likelihood ratio for an ARMA model of measurements. %while In our approach, the partial coherence is estimated through the empirical sample covariance for $\x, \y$ and $\z$ and is model-free.  
Our test statistic is the ordinary likelihood ratio for testing the null hypothesis that partial correlation is zero. It is model-free, requiring only the computation of partial coherence from sample correlations. %testing of computed from sample correlations, and is model-free.
To compare the two approaches, we replicate the experiment in \cite{Barnett1} by generating %the 
bivariate $ARMA(r ,1)$ time series data.  %The time series model diagram is illustrated in Figure.~\ref{fig:BarnettExpDiag}.
\begin{figure}[htbp] 
\centering
\includegraphics[width=0.4\linewidth]{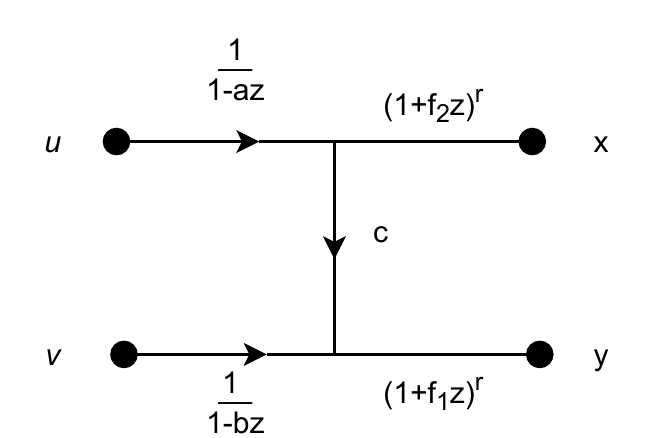}
\vspace{0.2in}
\includegraphics[width=.4\linewidth]{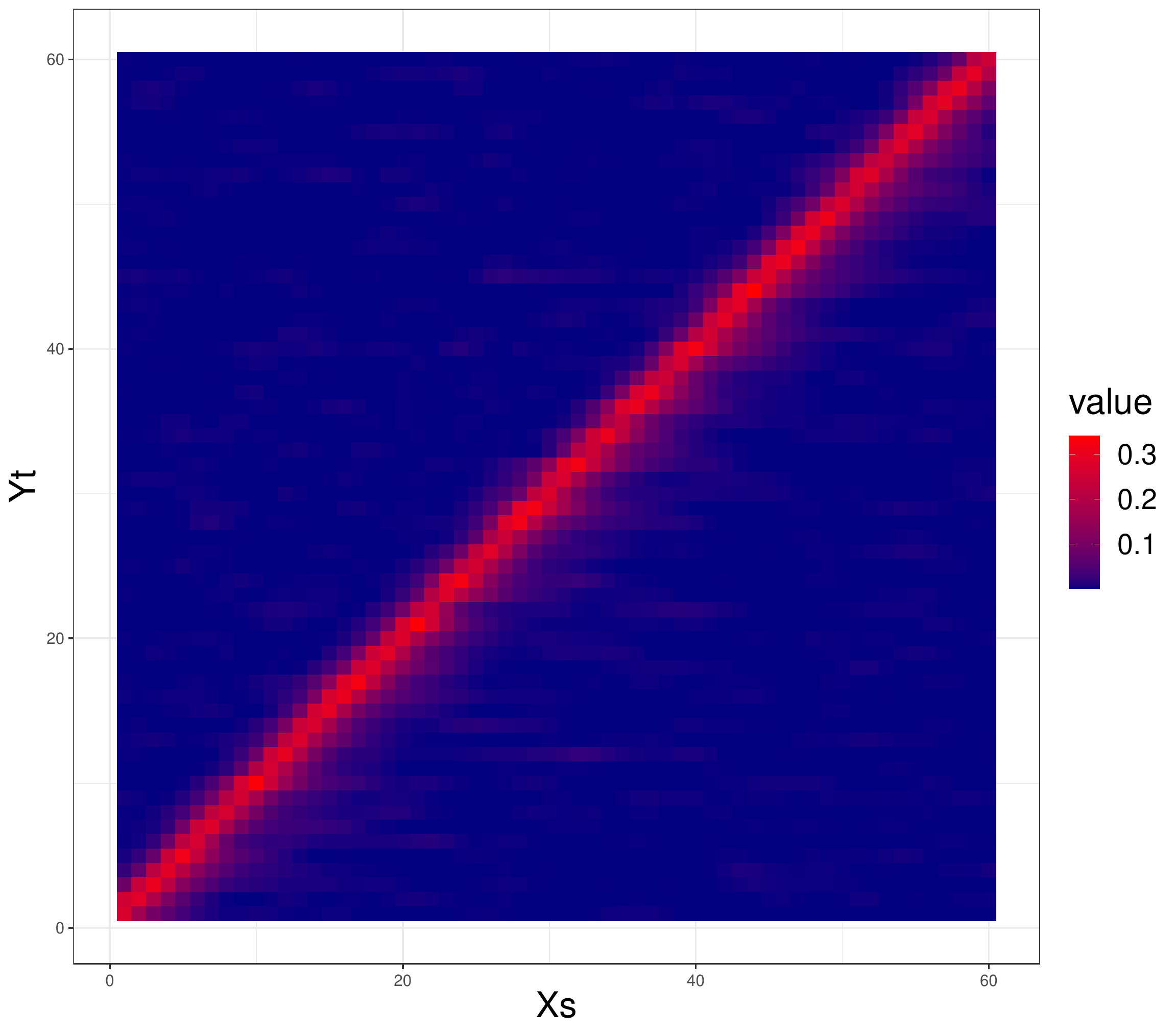}
\caption{Left: The time series model in the simulation study from Barnett et al. \cite{Barnett1} where $\mu$ and $\nu$ are time series of white noise. Right: Pairwise partial  coherence between $x_s$ at time $s$ ( horizontal axis) and $y_t$ at  time $t$ (vertical axis) with the prima facie evidence $\z$ is the past of $\{y_n\}$ up to time $t-1$. } %for each choice of $p$.} %with power versus the significance level}
\label{fig:BarnettExpDiag}
\end{figure}	

 Start with a bivariate auto-regressive (AR) time series $\{[\eta_{1t}, \eta_{2t}]^T\}$, %with
\begin{align*}
(1-az)\eta_{1t} &=%a \eta_{1,t-1}+
cz\eta_{2, t}+\mu_{t},\\
(1-bz)\eta_{2t} &= %b\eta_{2, t-1}+
\nu_{t},
\end{align*}
where $z$ is the backward shift or delay operator: $z\{\eta_t\}=\{\eta_{t-1}\}$.
%\begin{align*}
%\eta_{1t} &=a \eta_{1,t-1}+c\eta_{2, t-1}+\epsilon_{1t}\\
%\eta_{2t} &= b\eta_{2, t-1}+\epsilon_{2t}
%\end{align*}
The bivariate auto-regressive moving average (ARMA) time series %data 
are generated as %by an order $r$ backward operator $z$  as
\begin{eqnarray}
\left[\begin{array}{cc} y_{t} \\ x_{t}  \end{array}\right]= \left[\begin{array}{cc} (1+f_1z)^r & 0 \\ 0&(1+f_2z)^r  \end{array}\right] \cdot  \left[\begin{array}{cc} \eta_{1t} \\ \eta_{2t}  \end{array}\right] . \label{eqn:Barnett}
\end{eqnarray}
The circuit %time series model 
diagram in the left panel %is illustrated in 
of Figure.~\ref{fig:BarnettExpDiag}  shows clearly the dependence of the  time series $\{y_n\}$ on $\{x_n\}$ through the coupling parameter $c$.
The MA order $r$ varies from $0$ to $10$, but there are only two parameters, $f_1=0.6$ and $f_2=0.7$ determining these $r$ MA parameters. The AR parameters are set to $a=0.9, b=0.8$, $
	c=\sqrt{e^{-F}(e^F-1)(e^F-b^2)}$ with the transfer entropy $F=0.02$. The true partial coherence in this case is $\rho_{\x\y|\z}^2=1-e^{-0.02}\approx 0.02$ and $GG=e^{-0.02}$.  %c=\sqrt{\rho^2(1/(1-\rho^2)-b^2)}$ with $\rho^2=1-e^{-0.02}{\color{blue}\approx 0.02}$. 

To replicate the Barnett, et al. experiment, we choose %$T=1000$ 
$1000$ time samples %steps 
to estimate $\rho^2_{\x\y|\z}$ from second-order statistics, and % {\color{blue} $\rho^2_{\x\y|\z}$, %$\rho^2$, 
this experiment is replicated %are sampled in each of 
$10,000$ times to obtain Monte-Carlo estimates of performance.  The vector $\y$ is defined to be $y_t$;  $\x$ and $\z$ are truncated to finite  $T$-dimensional vectors with  $\x=[x_{t-1}, \ldots, x_{t-T}]$ and $\z=[y_{t-1}, \ldots, y_{t-T}]$.  So the sample covariance matrices for these vectors are just $10\times10$. Using the null distribution given in Theorem~1, we set the rejection region with significance level $0.05$.  As a practical matter, $(10,000)\times (1000)$ samples are generated and consecutive windows of $1000$ samples are constructed %for each of 
to obtain the 10,000 replications. Because of the weak dependence between windows, these are not actually $1000$ {\em independent} realizations, so our analytical formula for the null distribution does not perfectly predict the experimentally-determined significance level. But it is very close. In fact, when we design for a significance level of $\alpha=0.05$ we  achieve an experimental level of $0.044.$   

Figure~\ref{fig:simBarnett} displays our results based on the model-free test in Section~\ref{sec:test}.  The left panel plots the power of our test statistic with respect to the order of the underlying ARMA model, at significance level $0.05$. This power is relatively invariant to MA order, whereas the experimental results in \cite{Barnett1} show decreasing power as the MA order increases, as a result of the increased model complexity. % of their algorithm.  
 %The red curve shows that our 
The model-free coherence statistic for  memory $T=10$ has power about $0.90$ regardless of the MA order. This power is slightly lower than the result of Barnett, et al. at low MA orders, and higher at higher MA orders.   Power is relatively invariant to MA order, and it is not susceptible to model mismatch, as no state space model is assumed. %ption is made. 
%The  green ($k=20$) and blue ($k=30$)  curves 
%demonstrate that power drops as $k$ increase. Why? Because the partial coherence is estimated using the sample covariance matrices for $2k+1$ dimensional vector $(x, \y, \z)$. {\color{blue} There are $k(2k+1)/2$ covariance parameters to be estimated for the composite covariance $\R$, and the larger is $k$ the fewer samples there are per parameter. For example, at $k=30$, there are $915$ correlation parameters to be estimated,  leaving just $1000/915$ samples per parameter. } At $k=10$ this number is $1000/110$ samples per parameter.
 An appropriate value of $T$ can be selected  by cutting off  the empirical auto-correlations at a pre-determined threshold or comparing %the 
AIC or BIC values for an auto-regressive  model fitted to the sampled time series data.   In the right panel of  Figure~\ref{fig:simBarnett}, we plot the ROC curve (power vs size) to demonstrate that probability of correctly rejecting the null hypothesis (the power) depends on the probability of incorrectly rejecting it (the size). %with power (the true positive rate) against the false positive rate $\alpha$. It can be seen that the power increases very quickly as $\alpha$ increases. 
We also plot in the right panel of  Figure~\ref{fig:BarnettExpDiag}  the pairwise partial coherence between $y_t$ and $x_s$ given the  prima facie evidence $\z=[y_{t-1}, \ldots, y_{t-T}]$. These pairwise partial coherences are not the partial correlation $\rho^2_{\x\y|\z}$ that is used in the hypothesis test. Rather they are pairwise partial coherences, $\rho^2_{x_s,y_t|\z}$, plotted only for intuition,  %being the past of $\{y_n\}$. % and the past of $\{x_n\}$, respectively. 
%These partial coherence are estimated by the sample partial coherence with samples extracted from sliding window catch of the time series generated from~\eqref{eqn:Barnett}. 
for the case %For illustration purpose, we set the parameter 
$c=1$.   %In fact, $\rho^2_{\x\y|\z}$ is $\rho^2_{x_s,y_t|\z}$ integrated over finite time window.   
It can be seen that the partial coherence given the past of $\{y_n\}$ is non-zero %on 
for %both 
short ranges of $s>t$ and $s<t$.

\begin{figure}[htbp] 

\centering
\includegraphics[width=0.4\linewidth]{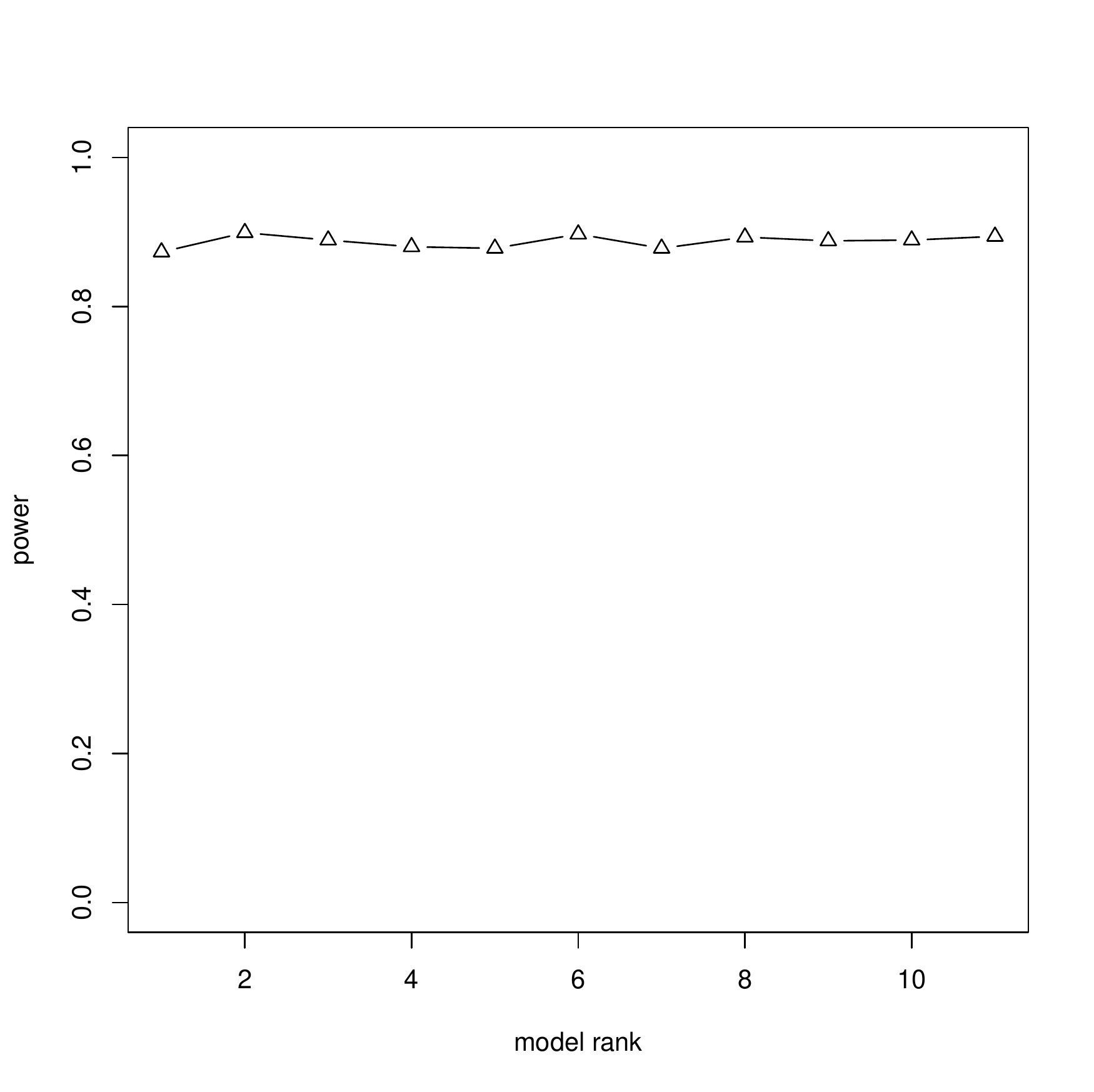}
\includegraphics[width=0.4\linewidth]{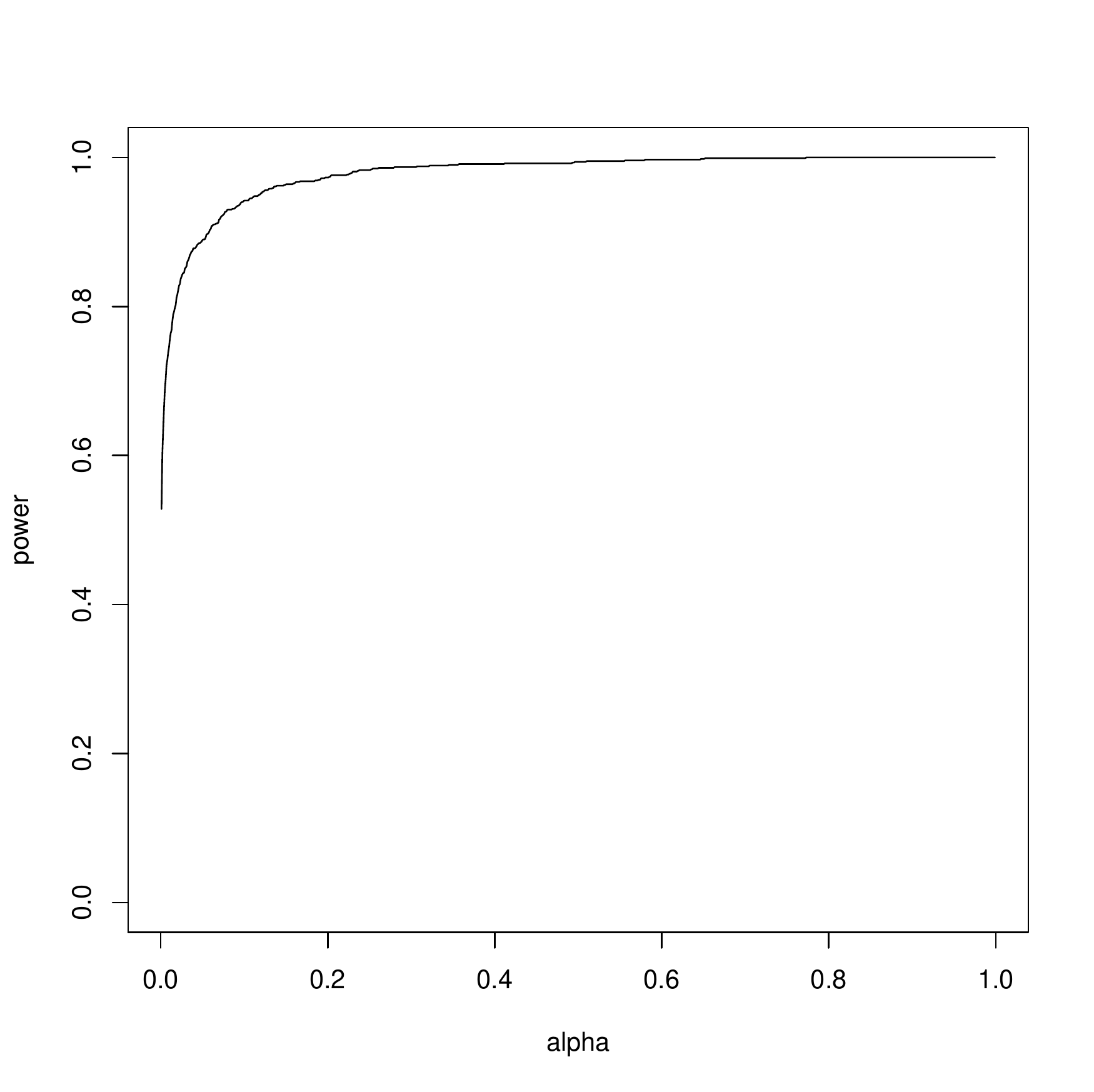}
\caption{The simulation study from Barnett et al. \cite{Barnett1}. (Left) Power versus the MA order of the ARMA time series. %model.  The red ($p=10$), green ($p=20$) and blue ($p=30$) curves represent three different choices for the conditioning on the past.%length of past. 
(Right) Receiver operating characteristic (ROC) showing power versus size.} %for each choice of $p$.} %with power versus the significance level}
\label{fig:simBarnett}
\end{figure}

\section{Discussion}\label{sec:discuss}

In this paper, our  objective has been to explore the use of partial coherence  among finite-dimensional vectors, constructed from  signal sequences, %time series, 
as a scale-invariant statistic to %using the partial coherence to 
assess %the %casual 
%temporal 
{\em causal influence of one signal sequence on another}.  %
%relation 
%between two time series.  
Of course these signal sequences might well be space series or space-time series, in which case the idea of past-, future- or mixed- causal influence is also a physically meaningful question \cite{Ebert-Uphoff:2012}. %The underlyng model for the time serieswe stu is linear. Our experimental study shows that this framework fits %well 
%within the paradigm of Grander Causality. %and it is capable of answering
%questions of causality from a given information set. 
The partial coherence statistic we argue for is a monotone function of  conditional KL divergence, conditional  mutual information,  conditional information rate, and the Granger-Geweke causality measure in the multivariate normal case. It is  a measure of second-order partial correlation in the Hilbert space of second-order random variables, and it is resolved by partial canonical correlations. The sampled-data version of partial  coherence  is in fact a likelihood ratio statistic whose null distribution %of this statistic and shown that it 
is a Wilks Lambda.  It  may be computed from the sample information matrix.  Importantly, from this null distribution, a threshold on partial coherence may be set to control the confidence level $1-\alpha$.

The key idea in the application of partial coherence to questions of  causal influence  is to define three channels of signals %measurements 
so that they code for the %causality 
question of interest. Then finite-dimensional vectors are constructed from these channels and their composite covariance matrix is estimated without appeal to an underlying generative model.  Thus the method is model-free.

We have demonstrated %this 
the methodology for two particular experiments in  causal influence. One of these experiments demonstrates how partial coherence may be used to map out a partial coherence surface, parameterized by two time variables, and the other replicates the experiment of  Barnett and Seth \cite{Barnett1}. For this latter experiment we demonstrate the power of our method and plot its experimentally computed Receiver Operating Characteristic (ROC). The conclusion of the test is that the performance of a model-free coherence test for  causal influence is competitive with a model-based coherence test. Moreover, the null distribution of the model-free test is invariant to the underlying generative model for the measurements. Consequently, as a null hypothesis test, the partial coherence test is not vulnerable to model bias. Of course in %some 
applications where a model might be well-established, and measurements are plentiful, a method like that of the Barnetts and Seth \cite{Barnett1}, \cite{Barnett2} might outperform a model-free method such as model-free partial coherence. In applications where there is no generative model like an ARMA model, then the method of second-order partial coherence is compelling. %Moreover, partial coherence may be applied without reference to time-, space-, or space-time-series,  making it applicable to quite general problems of  linear analysis for causal influence.

%We argue for a scale-invariant Moreover, we think it is the natural
%way to talk about causality for this set of problems and for this
%definition of causality, namely causality as indicated by linear
%conditional independence. While the existing methods of assessing causal relation exploit two regression models which are identical under the null, 
%partial coherence is a scale invariant scalar statistic whose null distribution is known under the null. 
{\bf Acknowledgment.}  We thank Profs. Sonia and Tomas Charleston Villalobos for piquing our interest in causality, and  Prof. Todd Moon for alerting us to the work of the Barnetts and Seth, which led us to the work of Geweke. \\

\bibliographystyle{IEEEbib}

 \end{document}